\documentclass[journal,onecolumn,12pt,twoside]{IEEEtran}
\IEEEoverridecommandlockouts

\usepackage{amsmath, amssymb, cite}
\usepackage{amsthm,bm,bbm}
\usepackage{mathtools}
\usepackage{amsthm,multirow,color,amsfonts}
\usepackage{tabulary}
\usepackage{subfigure}
\usepackage{graphicx}
\usepackage{setspace}
\usepackage{enumerate}
\usepackage{algorithm}
\usepackage{algpseudocode}
\usepackage{comment}
\usepackage{pdfpages}
\usepackage{hyperref,url}

\usepackage[size=small]{caption}

\usepackage{etoolbox}
\makeatletter\patchcmd{\@makecaption}  {\scshape}  {}  {}  {}\makeatother

\newcommand{\Deltaprod}{\xi_N}

\newtheorem{theo}{Theorem}
\newtheorem{defi}{Definition}

\newtheorem{prop}{Proposition}

\newtheorem{ex}{Example}

\def\compactify{\itemsep=0pt \topsep=0pt \partopsep=0pt \parsep=0pt}
\let\latexusecounter=\usecounter

\begin{document}

\title{Multi-Server Multi-Function Distributed Computation}
\author{Derya Malak, Mohammad Reza Deylam Salehi, Berksan Serbetci, and Petros Elia
\thanks{Derya Malak, Mohammad Reza Deylam Salehi, and Petros Elia are with the Commun. Systems Dept., EURECOM, Biot Sophia Antipolis, 06904 FRANCE (emails: \{malak,\, deylam, \, elia \}@eurecom.fr). This work was conducted when B. Serbetci was a Postdoctoral Researcher at EURECOM; fberks@gmail.com}
\thanks{This research was partially supported by a Huawei France-funded Chair towards Future Wireless Networks, and supported by the program ``PEPR Networks of the Future" of France 2030. 
Co-funded by the European Union (ERC, SENSIBILITÉ, 101077361, and ERC-PoC, LIGHT, 101101031). Views and opinions expressed are however those of the author(s) only and do not necessarily reflect those of the European Union or the European Research Council. Neither the European Union nor the granting authority can be held responsible for them.}}

\maketitle

\begin{abstract}
The work here studies the communication cost for a multi-server multi-task distributed computation framework, and does so for a broad class of functions and data statistics. Considering the framework where a user seeks the computation of multiple complex (conceivably non-linear) tasks from a set of distributed servers, we establish communication cost upper bounds for a variety of data statistics, function classes and data placements across the servers.    
To do so, we proceed to apply, for the first time here, K{\"o}rner's characteristic graph approach --- which is known to capture the structural properties of data and functions --- to the promising framework of multi-server multi-task distributed computing. Going beyond the general expressions, and in order to offer clearer insight, we also consider the well-known scenario of cyclic dataset placement and linearly separable functions over the binary field, in which case our approach exhibits considerable gains over the state of art. Similar gains are identified for the case of multi-linear functions.
\end{abstract}


\section{Introduction}
\label{sec:intro}

Distributed computing plays an increasingly significant role in accelerating the execution of computationally challenging and complex computational tasks. This growth in influence is rooted in the innate capability of distributed computing to parallelize computational loads across multiple servers. This same parallelization renders distributed computing as an indispensable tool for addressing a wide array of complex computational challenges, spanning scientific simulations, extracting various spatial data distributions \cite{yang2011using}, data-intensive analyses for cloud computing \cite{shamsi2013data}, machine learning \cite{10129894}, as well as applications in various other fields such as computational fluid dynamics \cite{gan1995evaluation}, high-quality graphics for movie and game rendering \cite{gao2019cost}, and a variety of medical applications \cite{lushbough2010overview} to name just a few. In the center of this ever-increasing presence of parallelized computing, stand modern parallel processing techniques, such as MapReduce \cite{dean2008mapreduce, 6903263, 8538364}, and Spark \cite{zaharia2010spark, khumoyun2016spark}.

%

For distributed computing though to achieve the desirable parallelization effect, there is an undeniable need for massive information exchange to-and-from the various network nodes. Reducing this communication load is essential for scalability \cite{orgerie2014survey, 10.1145/1142473.1142507,li2021flexible,liu2018distributed} in various topologies \cite{noormohammadpour2017datacenter, 179115, bestavros1995demand}. 
Central to the effort to reduce communication costs, stand coding techniques such as those found in \cite{8849431,8755563,tandon2017gradient, ye2018communication,8437467, MaddahAli2013, Karamchandani2016, LiSupAliAve2017,li2017fundamental, LiAliYuAves2018, YuAliAves2018, NadMadAve2016, subramaniam2019collaborative, dutta2019optimal, yosibash2021frame, DimGodWuWainRam2010, wan2022cache, jia2021capacity, soleymani2021analog}, including gradient coding \cite{tandon2017gradient}, and different variants of coded distributed computing that nicely yield gains in reliability, scalability, computation speed and cost-effectiveness~\cite{MaddahAli2013}.  
Similar communication-load aspects are often addressed via polynomial codes~\cite{YuAliAve2017} which can mitigate stragglers and enhance the recovery threshold, while MatDot codes, devised in \cite{dutta2019optimal,lopez2022secure} for secure distributed matrix multiplication, can decrease the number of transmissions for distributed matrix multiplication. This same emphasis in reducing communication costs is even more prominent in works like~\cite{wan2022cache, jia2021capacity, wan2021distributed,lopez2022secure, zhu2023information, 9478901,8949560, fawzi2022discovering, dutta2019optimal, aliasgari2020private, d2020notes, RasShaKum2011}, which again focus on distributed matrix multiplication. For example, focusing on a cyclic dataset placement model, the work in~\cite{wan2021distributed} provided useful achievability results, while the authors of \cite{jia2021capacity} have characterized achievability and converse bounds for secure distributed matrix multiplication. Furthermore, the work in~\cite{wan2022cache} found creative methods to exploit the correlation between the entries of the matrix product in order to reduce the cost of communication.


 %

\subsection{The multi-server multi-function distributed computing setting, and the need for accounting for general non-linear functions}
\label{sec:multi-server-function}
As computing requirements become increasingly challenging, distributed computing models have also evolved to be increasingly complex. One such recent model is the multi-server multi-function distributed computing model that consists of a master node, a set of distributed servers, and a user demanding the computation of multiple functions. The master contains the set of all datasets and allocates them to the servers which are then responsible for computing a set of specific subfunctions of datasets. 
This multi-server multi-function setting was recently studied by Wan {\em et al.} in \cite{wan2021distributed} for the class of linearly separable functions, which nicely captures a wide range of real-world tasks \cite{dean2008mapreduce} such as convolution \cite{9478901}, the discrete Fourier transform \cite{cances2023density} and a variety of other cases as well. This same work bounded the communication cost, employing linear encoding and linear decoding that leverage the structure of requests. 

At the same time though there is growing need to consider more general classes of functions, including non-linear functions such as is often the case with subfunctions that produce intermediate values in MapReduce operations~\cite{dean2008mapreduce}, or that relate to quantization \cite{hanna2020distributed}, classification \cite{luo2007distributed}, and optimization \cite{karakus2017straggler}. %
Intense interest can also be identified in the aforementioned problem of distributed matrix multiplication, which has been explored in a plethora of works that include~\cite{jia2021capacity, jia2021cross, wang2021price, chang2018capacity, d2020notes,8949560}, with a diverse focus that entails secrecy  \cite{d2020notes, chang2018capacity, jia2021cross}, as well as precision and stragglers  \cite{jia2021capacity, wang2021price, li2021flexible, 8949560} to name a few. 
In addition to matrix multiplication, other important non-linear function classes include sparse polynomial multiplication~\cite{monagan2009parallel}, permutation invariant functions \cite{hsu2021scalable} --- which often appear in multi-agent settings and have applications in learning, combinatorics, and graph neural networks --- as well as nomographic functions \cite{goldenbaum2014nomographic, goldenbaum2013harnessing} which can appear in the context of sensor networks and which have strong connections with interference exploitation and lattice codes, as nicely revealed in~\cite{goldenbaum2014nomographic, goldenbaum2013harnessing}. 

Our own work here is indeed motivated by this emerging need for distributed computing of non-linear functions, and our goal is to now consider general functions in the context of the multi-server, multi-function distributed computing framework, while also capturing dataset statistics and correlations, and while exploiting the structural properties of the (possibly non-linear) functions requested by the user. To do so, we will go beyond the linear coding approaches in~\cite{wan2021distributed, huang2023fundamental, 9678313}, and will devise demand-based encoding-decoding solutions. To do this, we will adopt --- in the context of the multi-server multi-function framework --- the powerful tools from characteristic graphs, which are specifically geared toward capturing both the statistical structure of data as well as the properties of functions beyond the linear case. To help the reader better understand our motivation and contribution, we proceed with a brief discussion on data structure and characteristic graphs. 

\subsection{Data correlation and structure}
\label{sec:char-graph}
Crucial in reducing the communication bottleneck of distributed computing, is an ability to capture the structure that appears in modern datasets. Indeed, even before computing considerations come into play, capturing the general structure of data has been crucial in reducing the communication load in various scenarios such as those in the seminal work by Slepian-Wolf \cite{SlepWolf1973} and Cover \cite{cover1975proof}.  
Similarly, when function computation is introduced, data structure can be a key component. In the context of computing, we have seen the seminal work by K{\"o}rner and Marton \cite{korner1979encode} which focused on efficient compression of the modulo $2$ sum of two statistically dependent sources, while Lalitha {\em et al.} \cite{lalitha2013linear} explored %
linear combinations of multiple statistically dependent sources. Furthermore, for general bivariate functions of correlated sources, when one of the sources is available as side information, the work of Yamamoto \cite{Yamamoto1982} generalized the pioneering work of Wyner and Ziv \cite{WynZiv1976}, to provide a rate-distortion characterization for the function computation setting. 

It is the case though that when the computational model becomes more involved --- as is the case in our multi-server multi-function scenario here --- data may often be treated as unstructured and independent~\cite{wan2020cache,wan2021distributed,huang2023fundamental,khalesi2022multi,9521491}. This naturally allows for crucial analytical tractability, but it may often ignore the potential benefits of accounting for statistical skews and correlations in data when aiming to reduce communication costs in distributed computing. Furthermore, this comes at a time when more and more function computation settings --- such as in medical imaging analysis \cite{erickson2017machine}, data fusion and group inferences~\cite{5484191}, as well as predictive modeling for artificial intelligence \cite{kant2015predictive} --- entail datasets with prominent dependencies and correlations.
While various works, such as by K\"orner-Marton \cite{korner1979encode}, Han-Kobayashi \cite{han1987dichotomy}, Yamamoto \cite{Yamamoto1982}, Alon-Orlitsky \cite{AlonOrlit1996}, Orlitsky-Roche \cite{OrlRoc2001}, provide crucial breakthroughs in exploiting data structure, to the best of our knowledge, \emph{in the context of fully distributed function computation}, the structure in functions and data has yet to be considered simultaneously.

\subsection{Characteristic graphs}
\label{sec:characteristic graph} 
To jointly account for this structure in both data and functions, we will draw from the powerful literature on \emph{Characteristic graphs}, introduced by K{\"o}rner for source coding \cite{Korner1973}, and used in data compression \cite{korner1979encode, AlonOrlit1996, OrlRoc2001, malak2022fractional, malak2023weighted, charpenay2023complementary}, cryptography \cite{salehi2023achievable}, image processing \cite{maugey2021graph}, and bioinformatics \cite{sevilla2005correlation}.
For example, toward understanding the fundamental limits of distributed functional compression, the work in \cite{Korner1973} devised the graph entropy approach in order to provide the best possible encoding rate of an information source with vanishing error probability. This same approach, while capturing both function structure and source structure, was presented for the case of one source, and it is not directly applicable to the distributed computing setting. Similarly the zero-error side information setting in \cite{AlonOrlit1996} and the lossy encoding setting in \cite{OrlRoc2001}, \cite{Yamamoto1982} use K\"orner's graph entropy \cite{Korner1973} approach to capture both function structure and source structure, but were again presented for the case of one source. Similar focus can be found in the works in~\cite{AlonOrlit1996,OrlRoc2001,malak2022fractional, malak2023weighted,salehi2023achievable}. The same characteristic graph approach has been nicely used by Feizi and M\'edard in \cite{FeiMed2014} for the distributed computing setting, for a simple distributed computing framework, and in the absence of considerations for data structure.

Characteristic graphs, which are used in fully distributed architectures to compress information, can allow us to capture various data statistics and correlations, various data placement arrangements, and various function types. This versatility motivates us to employ characteristic graphs in our multi-server, multi-function architecture for distributed computing of non-linear functions.

%

\subsection{Contributions}
\label{sec:contributions}
In this paper, leveraging fundamental principles from source and functional compression as well as graph theory, we study a general multi-server multi-function distributed computing framework composed of a single user requesting a set of functions, which are computed with the assistance of distributed servers that have partial access to datasets. To achieve our goal, we consider the use of K\"orner's characteristic graph framework~\cite{Korner1973} in our multi-server multi-function setting, and proceed to establish upper bounds on the achievable sum-rates reflecting the setting's communication requirements. 

By extending, for the first time here, K\"orner's characteristic graph framework~\cite{Korner1973} to the new multi-server multi-function setting, we are able to reflect the nature of the functions and data statistics, in order to allow each server to build a codebook of encoding functions that determine the transmitted information. Each server, using its own codebook, can transmit a function (or a set of functions) of the subfunctions of the data available in its storage, and to then provide the user with sufficient information for evaluating the demanded functions. The codebooks allow for a substantial reduction in the communication load. 

The employed approach allows us to account for general dataset statistics, correlations, dataset placement, and function classes, thus yielding gains over the state of art~\cite{wan2021distributed}, \cite{SlepWolf1973}, as showcased in our examples for the case of linearly separable functions in the presence of statistically skewed data, as well as for the case of multi-linear functions where the gains are particularly prominent, again under statistically skewed data. For this last case of multi-linear functions, we provide an upper bound on the achievable sum-rate (see Subsection~\ref{ex_prop:multi_shot_multilinear_function}), under a cyclic placement of data that reside in the binary field. 
We also provide a generalization of some elements in existing works on linearly separable functions~\cite{wan2021distributed,huang2023fundamental}.

In the end, our work demonstrates the power of using characteristic graph-based encoding for exploiting the structural properties of functions and data in distributed computing, as well as provides insights into fundamental compression limits, all for the broad scenario of multi-server, multi-function distributed computation.

%

%
%

%

%

%

\subsection{Paper organization}
\label{sec:org}
The rest of this paper is structured as follows. Section~\ref{model} describes the system model for the multi-server multi-function architecture, 
and Section~\ref{results} details the main results on the communication cost or sum-rate bounds under general dataset distributions and correlations, dataset placement models, and general function classes requested by the user, over a field of characteristic $q\geq 2$, through employing the characteristic graph approach, and contrasts the sum-rate with relevant prior works, e.g., \cite{wan2021distributed, SlepWolf1973}.
Finally, we summarize our key results and outline possible future directions in Section~\ref{sec:conclusion}. 
We provide a primer for the key definitions and results on characteristic graphs and their fundamental compression limits in Appendix~\ref{sec:Preliminary}, and give proofs of our main results in Appendix~\ref{sec:Proofs}.

{\bf Notation:} We denote by $H(X)=\mathbb{E}[-\log P_{X}(X)]$ the Shannon entropy of random variable $X$ drawn from distribution or probability mass function (PMF) $P_{X}$. 
Let $P_{X_1,X_2}$ be the joint  PMF of two random variables $X_1$ and $X_2$, where $X_1$ and $X_2$ are not necessarily independent and identically distributed (i.i.d.), i.e., equivalently the joint PMF is not in product form. 
The notation $X\sim {\rm Bern}(\epsilon)$ denotes that $X$ is Bernoulli distributed with parameter $\epsilon\in [0,1]$. 
Let $h(\cdot)$ denote the binary entropy function, and $H_B(B(n,\epsilon))$ denote the entropy of a Binomial random variable of size $n\in\mathbb{N}$, with $\epsilon\in[0,1]$ modeling the success probability of each Boolean-valued outcome. 
The notation $X_{\mathcal{S}}=\{X_i\,:\, i\in\mathcal{S}\}$ denotes a subset of servers with indices $i\in\mathcal{S}$ for $\mathcal{S}\subseteq\Omega$. The notation $\mathcal{S}^c=\Omega\backslash\mathcal{S}$ denotes the complement of $\mathcal{S}$. 
We denote the probability of an
event $A$ by $\mathbb{P}(A)$. 
The notation $1_{x\in A}$ denotes the indicator function which takes the value $1$ if $x\in A$, and $0$ otherwise. 
The notation $G_{X_i}$ denotes the characteristic graph that server $i\in\Omega$ builds for computing $F(X_{\Omega})$. 
The measures $H_{G_X}(X)$ and $H_{G_X}(X\,\vert\, Y)$ denote the entropy of characteristic graph $G_X$, and the conditional graph entropy for random variable $X$ given $Y$, respectively. 
The notation $\mathcal{T}(N, K, K_c, M, N_r)$ shows the topology of the distributed system. 
We note that $\mathcal{Z}_i$ denotes the indices of datasets stored in $i\in\Omega$, and the notation $K_n(\mathcal{S})=|\mathcal{Z}_{\mathcal{S}}|=\big|\bigcup\nolimits_{i\in \mathcal{S}} \mathcal{Z}_i\big|$  represents the cardinality of the datasets in the union of the sets in $\mathcal{S}$ for a given subset $\mathcal{S}\subseteq \Omega$ of servers. 
We also note that $[N]=\{1,2,\dots,N\}$, $N\in \mathbb{Z}^+$, and $[a:b]=\{a,a+1,\dots,b\}$ for $a,b\in\mathbb{Z}^+$ such that $a<b$. 
We use the convention $\mod{\{b,a\}}=a$ if $a$ divides $b$.  
We provide the notation in Table \ref{table:tab1}.

\begin{table*}[h!]
\footnotesize
\setlength{\extrarowheight}{1pt}
\begin{center}
\begin{tabular}{l | l }
%
%
{\bf Distributed computation system-related definitions} & {\bf Symbols}\\
\hline
Number of distributed servers; set of distributed servers; capacity of a server & $N$; $\Omega$; $M$ \\
Set of datasets; dataset catalog size & $\{D_k\}_{k\in[K]}$; $K=|\mathcal{K}|$\\
Subfunction $k\in \mathcal{Z}_i\subseteq [K]$ & $W_k=h_k(D_k)$ \\
The number of symbols in each $W_k$;  blocklength  & $L$ ; $n$\\
Set of indices of datasets assigned to server $i\in \Omega$ such that $|\mathcal{Z}_i|\leq M$ & $\mathcal{Z}_i\subseteq [K]$\\
Set of subfunctions corresponding to a subset of servers with indices $i\in\mathcal{S}$ for $\mathcal{S}\subseteq\Omega$ & $X_{\mathcal{S}}=\{X_i\,:\, i\in\mathcal{S}\}$\\
Recovery threshold & $N_r$\\
Number of demanded functions by the user & $K_c$\\
Number of symbols per transmission of server $i\in\Omega$ & $T_i$\\
Topology of the multi-server multi-function distributed computing setting  & $\mathcal{T}(N, K, K_c, M, N_r)$\\
& \\
{\bf Graph-theoretic definitions} & {\bf Symbols}\\
\hline
Characteristic graph that server $i$ builds for computing $F(X_{\Omega})$ & $G_{X_i}$, $i\in\Omega$\\
Union of characteristic graphs that server $i$ builds for computing $\{F_j(X_{\Omega})\}_{j\in [K_c]}$ & $G^{\cup}_{X_i}$, $i\in\Omega$\\
Maximal independent set (MIS); set of all MISs of $i\in\Omega$ & $U_1$; $S(G_{X_1})$ \\
A valid coloring of $G_{X_i}$ & $c_{G_{X_i}}$\\
$n$-th OR power graph; a valid coloring of the $n$-th OR power graph & $G^n_{{\bf X}_i}$; $c_{G^n_{{\bf X}_i}}({\bf X}_i)$\\ 
Characteristic graph entropy of $X_i$ & $H_{G_{X_i}}(X_i)$\\
Conditional characteristic graph entropy of $X_i$ such that $i\in\mathcal{S}$ given $X_{\mathcal{S}^c}$ & $H_{G_{X_i}}(X_i\,\vert\, X_{\mathcal{S}^c})$\\
\hline
\end{tabular}
\end{center}
\vspace{-0.3cm}
\caption{Notation.}
\label{table:tab1}
\end{table*}

\section{System model}
\label{model}

This section outlines our multi-server multi-function architecture and details our main technical contributions, namely the communication cost for the problem of distributed computing of general non-linear functions, and the cost for special instances of the computation problem under some simplifying assumptions on the dataset statistics, dataset correlations, placement, and the structures of functions.

In the multi-server, multi-function distributed computation framework, the master has access to the set of all datasets, and distributes the datasets across the servers. The total number of servers is $N$, and each server has a capacity of $M$. Communication from the master to the servers is allowed, whereas the servers are distributed and cannot collaborate. The user requests $K_c$ functions that could be non-linear. Given the dataset assignment to the servers, 
any subset of $N_r$ servers is sufficient to compute the functions 
requested.  
We denote by $\mathcal{T}(N, K, K_c, M, N_r)$ the topology for the described multi-server multi-function distributed computing setting, which we detail in the following.

\subsection{\bf Datasets, subfunctions, and placement into distributed servers} 
\label{sec:dataset_related_definitions}

There are $K$ datasets in total, each denoted by $D_k$, $k\in[K]$. 
Each distributed server $i\in \Omega=[N]$ with a capacity of $M$ is assigned a subset of datasets with indices $\mathcal{Z}_i\subseteq [K]$ such that $|\mathcal{Z}_i|=M$, where 
the assignments possibly overlap.

Each server computes a set of subfunctions $W_k=h_k(D_k)$ for $k\in \mathcal{Z}_i\subseteq [K]$, $i\in\Omega$. Datasets $\{D_k\}_{k\in[K]}$ could be dependent \footnote{We note that exploiting the temporal and spatial variation or dependence of data, it is possible to decrease the communication cost.
} across $\mathcal{K}$, so could $\{W_k\}_{k\in[K]}$. 
We denote the number of symbols in each $W_k$ by $L$, which equals the blocklength $n$. 
Let $X_i=\{W_k\}_{k\in \mathcal{Z}_i}=W_{\mathcal{Z}_i}=\{h_k(D_k)\}_{k\in\mathcal{Z}_i}$ denote the set of subfunctions of $i$-th server, $\mathcal{X}_i$ be the alphabet of $X_i$, and $X_{\Omega}=(X_1,X_2,\dots,X_N)$ be the set of subfunctions of all servers. 
We denote by ${\bf W}_k=W_{k1},W_{k2},\dots,W_{kn}$ and ${\bf X}_i=X_{i1},X_{i2},\dots, X_{in}\in \mathbb{F}_q^{|\mathcal{Z}_i|\times n}$, the length $n$ sequences of subfunction $W_k$,  
and of 
$W_{\mathcal{Z}_i}$ assigned to server $i\in\Omega$.

\subsection{\bf Cyclic dataset placement model, computation capacity, and recovery threshold}
\label{sec:cyclic_placement}

We assume that the total number of datasets $K$ is divisible by the number of servers $N$, i.e., $\frac{K}{N}\doteq\Delta\in\mathbb{Z}^+$. 
The dataset placement on $N$ distributed servers is conducted in a circular or cyclic manner, in the amount of $\Delta$ circular shifts between two consecutive servers, where the shifts are to the right and the final entries are moved to the first positions, if necessary. 
As a result of cyclic placement, any subset of $N_r$ servers covers the set of all datasets to compute the requested functions from the user. 
Given $N_r\in [N]$, each server has a storage size or computation cost of $|\mathcal{Z}_i|=M=\Delta(N-N_r+1)$, and the amount of dataset overlap between the consecutive servers is $\Delta(N-N_r)$. 

Hence, the set of indices assigned to server $i\in\Omega$ 
is given as follows:
\begin{align}
\label{cyclic_Zi}
\mathcal{Z}_i = \bigcup_{r=0}^{\Delta-1} \left\{ \mod{\{i,N\}}+rN, \mod{\{i+1,N\}}+rN, \dots,\mod{\{i+N-N_r,N\}}+rN \,\right\},
\end{align}

where $X_i=W_{\mathcal{Z}_i}$, $i\in\Omega$. 
As a result of (\ref{cyclic_Zi}), the cardinality of the datasets assigned to each server meets the storage capacity constraint $M$ with equality, i.e., 
$|\mathcal{Z}_i|=M$, for all $i\in\Omega$.

\subsection{\bf User demands and structure of the computation} 
\label{sec:user_demands}

We address the problem of distributed lossless compression of a set general multi-variable functions $F_j(X_{\Omega}):\mathcal{X}_1\times \mathcal{X}_2 \dots\times \mathcal{X}_N\to \mathbb{F}_q$, $j\in [K_c]$, requested by the user from the set of servers, where $K_c\geq 1$, and the functions are known by the servers and the user. More specifically, the user, from a subset of distributed servers, aims to compute in a lossless manner the following length $n$ sequence as $n$ tends to infinity:
\begin{align}
\label{eq-general-nonlinear-func}
F_j({\bf X}_{\Omega})=\{F_j(X_{1l},\,X_{2l},\dots, X_{Nl})\}_{l=1}^n \, \, j\in [K_c] \ ,
\end{align}
where $F_j(X_{1l},\,X_{2l},\dots, X_{Nl})$ is the function outcome for the $l$-th realization $l\in [n]$, given the length $n$ sequence. We note that the representation in (\ref{eq-general-nonlinear-func}) is the most general form of a (conceivably non-linear) multi-variate function, which encompasses the special cases of separable functions, and linearly separable functions, which we discuss next.

In this work, the user seeks to compute functions that are separable to each dataset. Each demanded function $f_j(\cdot)\in\mathbb{R}$, $j\in[K_c]$ is a function of subfunctions $\{W_k\}_{k\in\mathcal{K}}$ such that $W_k=h_k(D_k)\in\mathbb{F}_q$, where $h_k$ is a general function (could be linear or non-linear) of dataset $D_k$. Hence, using the relation $X_i=W_{\mathcal{Z}_i}=\{h_k(D_k)\}_{k\in\mathcal{Z}_i}$, each demanded function $j\in[K_c]$ can be written in the following form: 
\begin{align}
f_j(W_{\mathcal{K}})
=f_j(h_1(D_1),\dots,h_K(D_K))
=F_j(\{h_k(D_k)\}_{k\in\mathcal{Z}_1},\dots,\{h_k(D_k)\}_{k\in\mathcal{Z}_N})
=F_j(X_{\Omega})%
\ .
\end{align}

In the special case of linearly separable functions\footnote{Special instances of the linearly separable representation of subfunctions $\{W_k\}_k$ given in (\ref{linearly_separable_functions}) are linear functions of the datasets $\{D_k\}$ and are denoted by $F_j=\sum_{k}\gamma_{jk} D_k$.} \cite{wan2021distributed}, the demanded functions take the form: 
\begin{align}
\label{linearly_separable_functions}
\{F_j(X_{\Omega})\}_{j\in [K_c]}
=
\begin{bmatrix}
F_1 & F_2 & \hdots & F_{K_c} 
\end{bmatrix}^{\intercal}
={\bf \Gamma}{\bf W}\ ,
\end{align}
where ${\bf W}=\begin{bmatrix}
W_1 & W_2 & \hdots & W_K 
\end{bmatrix}^{\intercal}\in\mathbb{F}_q^{K\times 1}$ is the subfunction vector, and the coefficient matrix ${\bf \Gamma}=\{\gamma_{jk}\}\in\mathbb{F}_q^{K_c\times K}$ is known to the master node, servers, and the user. In other words, $\{F_j(X_{\Omega})\}_{j\in[K_c]}$ is a set of linear maps from the subfunctions $\{W_k\}_k$, where $F_j(X_{\Omega})=\sum\nolimits_{k\in[K]} \gamma_{jk}\cdot W_k$. 
We do not restrict $\{F_j(X_{\Omega})\}_{j\in [K_c]}$ to linearly separable functions, i.e., it may hold that $\{F_j(X_{\Omega})\}_{j\in [K_c]}\neq {\bf \Gamma}{\bf W}$.

\subsection{\bf Communication cost for the characteristic graph-based computing approach}
\label{sec:communication_cost_char_graph}

To compute $\{F_j({\bf X}_{\Omega})\}_{j\in[K_c]}$, each server $i\in\Omega$ constructs a characteristic graph, denoted by $G_{X_i}$, for compressing $X_i$. More specifically, for asymptotic lossless computation of the demanded functions, the server builds the $n$-th OR power $G^n_{{\bf X}_i}$ of $G_{X_i}$ for compressing ${\bf X}_i$ 
to determine the transmitted information. 
The minimal possible code rate achievable to distinguish the edges of $G^n_{{\bf X}_i}$ as $n\to\infty$, is given the {\emph{Characteristic graph entropy}}, $H_{G_{X_i}}(X_i)$. 
For a primer on key graph-theoretic concepts, characteristic graph-related definitions, and the fundamental compression limits of characteristic graphs, we refer the reader to \cite{malak2022fractional}, \cite{salehi2023achievable}, \cite{FeiMed2014}.  
In this work, we solely focus on the characterization of the total communication cost from all servers to the user, i.e., the achievable sum-rate, without accounting for the costs of communication between the master and the servers, and of computations performed at the servers and the user.

Each $i\in\Omega$ builds a mapping from ${\bf X}_i$ to a valid coloring of $G^n_{{\bf X}_i}$, denoted by $c_{G^n_{{\bf X}_i}}({\bf X}_i)$. 
The coloring $c_{G^n_{{\bf X}_i}}({\bf X}_i)$ specifies the color classes of ${\bf X}_i$ that form independent sets to distinguish the demanded function outcomes. 
Given an encoding function $g_i$ that models the transmission of server $i\in\Omega$ for computing $\{F_j({\bf X}_{\Omega})\}_{j\in[K_c]}$, we denote by ${\bf Z}_i=g_i({\bf X}_i)=e_{X_i}(c_{G^n_{{\bf X}_i}}({\bf X}_i))$ the color encoding performed by server $i\in\Omega$ for ${\bf X}_i$. Hence, the communication rate of server $i\in\Omega$, for a sufficiently large blocklength $n$, where $T_i$ is the length %
for the color encoding performed at $i\in\Omega$, is 
\begin{align}
\label{communication_cost_functional_compression}
\!\! R_i=\frac{T_i}{L}=\frac{H(e_{X_i}(c_{G^n_{{\bf X}_i}}({\bf X}_i)))}{n}\geq H_{G_{X_i}}(X_i) \ ,\,\, i\in\Omega\ ,
\end{align}
where the inequality follows from exploiting the achievability of 
$H_{G_{X_i}}(X_i)=\lim\limits_{n\to\infty} \frac{1}{n}H^{\chi}_{G^n_{{\bf X}_i}}({\bf X}_i)$, where $H^{\chi}_{G^n_{{\bf X}_i}}({\bf X}_i)$ is the {\emph{chromatic entropy}} of the graph $G^n_{{\bf X}_i}$ \cite{AlonOrlit1996}, \cite{Korner1973}. We refer the reader to Appendix~\ref{subsec:Preliminary-graph} for a detailed description of the notions of chromatic and graph entropies (cf.~ (\ref{chromatic_entropy_expression}) and (\ref{Korners_graph_entropy}), respectively).

For the multi-server multi-function distributed setup, using the characteristic graph-based fundamental limit in (\ref{communication_cost_functional_compression}), an achievable sum-rate for asymptotic lossless computation is 
\begin{align}
R_{\rm ach}=\sum\limits_{i\in \Omega}R_i\leq \sum\limits_{i\in \Omega}H_{G_{X_{i}}}( X_{i}) \ . 
\end{align}

We next provide our main results in Section~\ref{results}.

\section{Main results}
\label{results}

In this section, we analyze the multi-server multi-function distributed computing framework exploiting the characteristic graph-based approach in \cite{Korner1973}. In contrast to the previous research attempts in this direction, our solution method is general, and it captures
(i) general input statistics or dataset distributions or the skew in data instead of assuming uniform distributions, 
(ii) correlations across datasets, (iii) any dataset placement model across servers, beyond the cyclic \cite{wan2021distributed} or the Maddah-Ali and Niesen \cite{MaddahAli2013Journal} placements, 
and
(iv) general function classes requested by the user, instead of focusing on a particular function type (see e.g., \cite{wan2021distributed,khalesi2022multi, 4544980}).

Subsequently, we will delve into specific function computation scenarios. First, we will present our main result (Theorem \ref{theo_cyclic_placement_general_source-general_function_graph-rate_UB}) which is the most general form that captures (i)-(iv). 
We then demonstrate (in Proposition~\ref{prop_cyclic_placement_iid_uniform_source-linear_function_graph-rate_UB}) that the celebrated result of Wan {\em et al.} \cite[Theorem 2]{wan2021distributed} can be obtained as a special case of Theorem \ref{theo_cyclic_placement_general_source-general_function_graph-rate_UB}, given that: 
(i) the datasets are i.i.d. and uniform over $q$-ary fields, (ii) the placement of datasets across servers is cyclic, and (iii) the demanded functions are linearly separable, given as in (\ref{linearly_separable_functions}). 
Under correlated and identically distributed Bernoulli dataset model with a skewness parameter~$\epsilon\in (0,1)$ for datasets, we next present in Proposition~\ref{prop:general_placement__correlated_Boolean_function}, the achievable sum rate for computing Boolean functions. 
Finally, in Proposition~\ref{prop:multi_shot_multilinear_function}, we analyze our characteristic graph-based approach for evaluating multi-linear functions, a pertinent class of non-linear functions, under the assumption of cyclic placement and i.i.d. Bernoulli distributed datasets with parameter~$\epsilon$, and derive an upper bound on the sum rate needed. To gain insight into our analytical results and demonstrate the savings in the total communication cost, we provide some numerical examples.  

We next present our main theorem (Theorem~\ref{theo_cyclic_placement_general_source-general_function_graph-rate_UB}) on the achievable communication cost for the multi-server, multi-function topology, which holds for all input statistics, under any correlation model across datasets, and for distributed computing of all function classes requested by the user, regardless of the data assignment over the servers' caches. The key to capturing the structure of general functions in Theorem~\ref{theo_cyclic_placement_general_source-general_function_graph-rate_UB} is the utilization of a characteristic graph-based compression technique, as proposed by K\"orner in \cite{Korner1973}. 
\footnote{For a more detailed description of characteristic graphs and their entropies, see Appendix~\ref{subsec:Preliminary-graph}.}

\begin{theo}
\label{theo_cyclic_placement_general_source-general_function_graph-rate_UB}    
{\bf (Achievable sum-rate using the characteristic graph approach for general functions and distributions.)} 
In the multi-server, multi-function distributed computation model, denoted by $\mathcal{T}(N, K, K_c, M, N_r)$, under general placement of datasets, and for a set of $K_c$ general functions $\{f_j(W_{\mathcal{K}})\}_{j\in[K_c]}$ requested by the user, and under general jointly distributed dataset models, including non-uniform inputs and allowing correlations across datasets, the characteristic graph-based compression yields the following upper bound on the achievable communication rate:
\begin{align}
\label{eq:general_communication_rate_upper_bound}
R_{\rm ach} \leq \sum\limits_{i=1}^{N_r}\, \min\limits_{Z_i=g_i(X_i)\, : \, g_i\in \mathcal{C}_i} H_{G^{\cup}_{X_i}}(X_i) \ ,
\end{align}
where 
\begin{itemize}
\item $G^{\cup}_{X_i}=\bigcup\limits_{j\in[K_c]} G_{X_i,j}$ is the union characteristic graph\footnote{We refer the reader to (\ref{union_graph}) (Appendix~\ref{subsec:Preliminary-graph}) for the definition of a union of characteristic graphs.} that server $i\in\Omega$ builds for computing $\{f_j(W_{\mathcal{K}})\}_{j\in[K_c]}$,
\item $\mathcal{C}_i\ni g_i$ denotes a codebook of functions that server $i\in\Omega$ uses for computing $\{f_j(W_{\mathcal{K}})\}_{j\in[K_c]}$,  
\item each subfunction $W_k$, $k\in\mathcal{K}$ is defined over a $q$-ary field such that the characteristic is at least $2$, 
and
\item $Z_i=g_i(X_i)$ such that $g_i\in\mathcal{C}_i$ denotes the transmitted information from server $i\in\Omega$.
\end{itemize}
\end{theo}

\begin{proof}
 See Appendix~\ref{Proof_theo_cyclic_placement_general_source-general_function_graph-rate_UB}.   
\end{proof}

Theorem~\ref{theo_cyclic_placement_general_source-general_function_graph-rate_UB} provides a general upper bound on the sum-rate for computing functions for general dataset statistics and correlations, and the placement model, and allows any function type, over a field of characteristic $q\geq 2$. We note that in (\ref{eq:general_communication_rate_upper_bound}), the codebook $\mathcal{C}_i$ determines the structure of the union characteristic graph $G^{\cup}_{X_i}$, which, in turn, determines the distribution of $Z_i$. Therefore, the tightness of the rate upper bound relies essentially on the codebook selection. 
We also note that it is possible to analyze the computational complexity of building a characteristic graph and computing the bound in (\ref{eq:general_communication_rate_upper_bound}) via evaluating the complexity of the transmissions $Z_i$ determined by $\{f_j(W_{\mathcal{K}})\}_{j\in[K_c]}$ for a given $i\in\Omega$. However, the current manuscript focuses primarily on the cost of communication, and we leave the computational complexity analysis to future work.
Because (\ref{eq:general_communication_rate_upper_bound}) is not analytically tractable, in the following, we will focus on special instances of Theorem~\ref{theo_cyclic_placement_general_source-general_function_graph-rate_UB}, to gain insights into the effects of input statistics, dataset correlations, and special function classes in determining the total communication cost. 

We next demonstrate that the achievable communication cost for the special scenario of {\em distributed linearly separable computation} framework given in \cite[Theorem 2]{wan2021distributed} is embedded by the characterization provided in Theorem~\ref{theo_cyclic_placement_general_source-general_function_graph-rate_UB}. We next showcase the achievable sum rate result for linearly separable functions. 

\begin{prop}
\label{prop_cyclic_placement_iid_uniform_source-linear_function_graph-rate_UB} 
{\bf (Achievable sum-rate using the characteristic graph approach for linearly separable functions and i.i.d.  subfunctions over $\mathbb{F}_q$.)} 
In the multi-server, multi-function distributed computation model, denoted by $\mathcal{T}(N, K, K_c, M, N_r)$, under the cyclic placement of datasets, where $\frac{K}{N}=\Delta\in\mathbb{Z}^+$, and for a set of $K_c$ linearly separable functions, given as in (\ref{linearly_separable_functions}), requested by the user, and given i.i.d. 
uniformly distributed subfunctions over a field of characteristic $q\geq 2$,  
the characteristic graph-based compression yields the following bound on the achievable communication rate:
\begin{align}
\label{eq:linearly_separable_communication_rate_upper_bound}
R_{\rm ach} \leq \begin{cases}
\min\{K_c,\Delta\}N_r\ ,\quad 1\leq K_c\leq \Delta Nr \ ,\\
\min\{K_c, K\} \ ,\quad \Delta Nr<K_c \ .
\end{cases}
\end{align}
\end{prop}

\begin{proof}
See Appendix~\ref{Proof_prop_cyclic_placement_iid_uniform_source-linear_function_graph-rate_UB}.
\end{proof}

We note that Theorem~\ref{theo_cyclic_placement_general_source-general_function_graph-rate_UB} results in Proposition~\ref{prop_cyclic_placement_iid_uniform_source-linear_function_graph-rate_UB} when three conditions hold: (i) the dataset placement across servers is cyclic following the rule in (\ref{cyclic_Zi}), (ii) the subfunctions $W_{\mathcal{K}}$ are i.i.d. and uniform over $\mathbb{F}_q$ (see (\ref{M_iid_W_k_variables}) in  Appendix~\ref{Proof_prop_cyclic_placement_iid_uniform_source-linear_function_graph-rate_UB}), and (iii) the codebook $\mathcal{C}_i$ is restricted to linear combinations of subfunctions $W_{\mathcal{K}}$, which yields that the independent sets of $G^{\cup}_{X_i}$ satisfy a set of linear constraints\footnote{We detail these linear constraints in Appendix~\ref{Proof_prop_cyclic_placement_iid_uniform_source-linear_function_graph-rate_UB}, where the set of linear equations given in (\ref{linear_encoding}) is used to simplify the entropy $H_{G^{\cup}_{X_i}}(X_i)$ of the union characteristic graph $G^{\cup}_{X_i}$ via the expression given in (\ref{mutual_info_definition_indep_set_for_Kc_functions}) for evaluating the upper bound given in (\ref{eq:general_communication_rate_upper_bound_app}) on the achievable sum rate for computing the desired functions via exploiting the entropies of the union characteristic graphs for each of the $N_r$ servers, given the recovery threshold $N_r$.} in the variables $\{W_k\}_{k\in\mathcal{Z}_i}$. 
Note that the linear encoding and decoding approach for computing linearly separable functions, proposed by Wan {\em et al.} in   \cite[Theorem~2]{wan2021distributed}, is valid over a field of characteristic $q>3$. However, in Proposition~\ref{prop_cyclic_placement_iid_uniform_source-linear_function_graph-rate_UB}, the characteristic of $\mathbb{F}_q$ is at least $2$, i.e., $q\geq 2$,  generalizing \cite[Theorem~2]{wan2021distributed} to larger input alphabets.

Next, we aim to demonstrate the merits of the characteristic graph-based compression in {\em capturing dataset correlations} within the multi-server, multi-function distributed computation framework. 
More specifically, we restrict the general input statistics in Theorem~\ref{theo_cyclic_placement_general_source-general_function_graph-rate_UB} such that the datasets are correlated and identically distributed, where each subfunction follows a Bernoulli distribution with the same parameter $\epsilon$, i.e., $W_k\sim {\rm Bern}(\epsilon)$, with $\epsilon\in (0,1)$, and the user demands $K_c$ arbitrary  Boolean functions regardless of the data assignment. Similarly to Theorem~\ref{theo_cyclic_placement_general_source-general_function_graph-rate_UB}, the following proposition (Proposition~\ref{prop:general_placement__correlated_Boolean_function}) holds for general function types (Boolean) regardless of the data assignment.

\begin{prop}
\label{prop:general_placement__correlated_Boolean_function}    
{\bf (Achievable sum-rate using the characteristic graph approach for general functions and identically distributed subfunctions over $\mathbb{F}_2$.)} 
In the multi-server multi-function distributed computing setting, denoted by $\mathcal{T}(N, K, K_c, M, N_r)$,
under the general placement of datasets, and for a set of $K_c$ Boolean functions $\{f_j(W_{\mathcal{K}})\}_{j\in[K_c]}$ requested by the user, and given identically distributed and correlated 
subfunctions with $W_k\sim {\rm Bern}(\epsilon)$, $k\in[K]$, where $\epsilon\in (0,1)$, the characteristic graph-based compression yields the following bound on the achievable communication rate: 
\begin{align}
\label{eq:bernouilli_correlated_subfunctions_communication_rate_upper_bound}
R_{\rm ach} \leq \sum\limits_{i=1}^{N_r}\, \min\limits_{Z_i=g_i(X_i)\, : \, g_i\in \mathcal{C}_i} h(Z_i) \ ,
\end{align}
where
\begin{itemize} 
\item $\mathcal{C}_i \ni g_i\,:\, \{0,1\}^{M}\to \{0,1\}$ denotes a codebook of Boolean functions that server $i\in\Omega$ uses,   
\item $Z_i=g_i(X_i)$ such that $g_i\in\mathcal{C}_i$ denotes the transmitted information from server $i\in\Omega$,
\item $G^{\cup}_{X_i}$ has two maximal independent sets (MISs), namely $s_0(G^{\cup}_{X_i})$ and $s_1(G^{\cup}_{X_i})$, yielding
$Z_i=0$ and $Z_i=1$, respectively, 
and 
\item the probability that $W_{\mathcal{Z}_i}$ yields the function value $Z_i=1$ is given as
\begin{align}
\label{prob_indep_set_1}
\mathbb{P}(Z_i=1)=\mathbb{P}(W_{\mathcal{Z}_i}\in s_1(G^{\cup}_{X_i}))\ , \quad i\in \Omega \ .
\end{align}
\end{itemize}

\end{prop}

\begin{proof}
See Appendix \ref{proof-prop:general_placement__correlated_Boolean_function}.
\end{proof}

While admittedly the above approach (Proposition~\ref{prop:general_placement__correlated_Boolean_function}) may not directly offer sufficient insight, it does employ the new machinery to offer a generality that allows us to plug in any set of parameters to determine the achievable performance.

Contrasting Propositions~\ref{prop_cyclic_placement_iid_uniform_source-linear_function_graph-rate_UB}-~\ref{prop:general_placement__correlated_Boolean_function}, which give the total communication costs for computing linearly separable and Boolean functions, respectively, over $\mathbb{F}_2$, Proposition~\ref{prop:general_placement__correlated_Boolean_function}, 
by exploiting the skew and correlations of datasets indexed by $\mathcal{Z}_i$, as well as the functions' structures via the MISs $s_0(G^{\cup}_{X_i})$ and $s_1(G^{\cup}_{X_i})$ of server $i\in\Omega$, demonstrates that harnessing the correlation across the datasets can indeed reduce the total communication cost versus the setting in Proposition~\ref{prop_cyclic_placement_iid_uniform_source-linear_function_graph-rate_UB}, devised with the assumption of i.i.d. and uniformly distributed subfunctions.

The prior works have focused on devising distributed computation frameworks and exploring their communication costs for specific function classes. For instance, in \cite{ korner1979encode},  K{\"o}rner and Marton have restricted the computation to be the binary sum function, and in \cite{han1987dichotomy}, Han and Kobayashi have classified functions into two categories depending on whether they can be computed at a sum rate that is lower than that of \cite{SlepWolf1973}. 
Furthermore, the computation problem has been studied for specific topologies, e.g., the side information setting in \cite{AlonOrlit1996,OrlRoc2001}. 
Despite the existing efforts, see e.g., \cite{korner1979encode,han1987dichotomy,AlonOrlit1996,OrlRoc2001}, to the best of our knowledge, for the given multi-server, multi-function distributed computing scenario, there is still no general framework for determining the fundamental limits of the total communication cost for computing general non-linear functions. Indeed, for this setting, the most pertinent existing work that applies to general non-linear functions and provides an upper bound on the achievable sum rate 
is that of Slepian-Wolf \cite{SlepWolf1973}. On the other hand, the upper bound on the achievable computation scheme presented in Theorem~\ref{theo_cyclic_placement_general_source-general_function_graph-rate_UB} can provide savings in the communication cost over \cite{SlepWolf1973} for functions including linearly separable functions and beyond. 
To that end, we exploit Theorem~\ref{theo_cyclic_placement_general_source-general_function_graph-rate_UB} to determine an upper bound on the achievable sum-rate for distributed computing of a multi-linear function in the form of
\begin{align}
\label{multilinear_function}  f(W_{\mathcal{K}})=\prod\limits_{k\in [K]} W_k \ .
\end{align}

Note that (\ref{multilinear_function}) is used in various scenarios, including distributed machine learning, e.g., to reduce variance in noisy datasets via ensemble learning  \cite{kaur2020comparison} or weighted averaging \cite{chen2021federated}, 
sensor network applications to aggregate readings for improved data analysis \cite{zhao2003computing}, as well as distributed optimization and financial modeling, where these functions play pivotal roles in establishing global objectives and managing risk and return \cite{giselsson2018large, kavadias2007resource}.

Drawing on the utility of characteristic graphs in capturing the structures of data and functions, as well as input statistics and correlations, and the general result in Theorem~\ref{theo_cyclic_placement_general_source-general_function_graph-rate_UB}, our next result,    Proposition~\ref{prop:multi_shot_multilinear_function}, demonstrates a new upper bound on the achievable sum rate for computing {\emph{multi-linear functions}} within the framework of multi-server and multi-function distributed computing via exploiting conditional graph entropies.

\begin{prop}
\label{prop:multi_shot_multilinear_function}
{\bf (Achievable sum-rate using the characteristic graph approach for multi-linear functions and i.i.d. subfunctions over $\mathbb{F}_2$.)} 
In multi-server multi-function distributed computing setting, denoted by $\mathcal{T}(N, K, K_c, M, N_r)$, under the cyclic placement of datasets, where $\frac{K}{N}=\Delta\in\mathbb{Z}^+$, and for computing the multi-linear function ($K_c=1$), given as in (\ref{multilinear_function}), requested by the user, and given i.i.d. uniformly distributed subfunctions $W_k\sim{\rm Bern}(\epsilon)$, $k\in[K]$, for some $\epsilon\in (0,1)$, the characteristic graph-based compression yields the following bound on the achievable communication rate:
\begin{align}
\label{sum_rate_multishot_multilinear}
R_{\rm ach}\leq \frac{1-(\epsilon_M)^{N^*}}{1-\epsilon_M}\cdot h(\epsilon_M)+(\epsilon_M)^{N^*}\cdot\, 1_{\Delta_N>0}\cdot\, h\big(\epsilon_{\Deltaprod}\big) \ ,
\end{align}
where 
\begin{itemize}
\item $\epsilon_M=\epsilon^M$ denotes the probability that the product of $M$ subfunctions, with $W_k\sim{\rm Bern}(\epsilon)$ being i.i.d. across $k\in [K]$, taking the value one, i.e., $\mathbb{P}\Big(\prod\nolimits_{k\in \mathcal{S}\,:\, |\mathcal{S}|=M} W_k\Big)=\epsilon_M$, 
\item the variable $N^*=\left\lfloor \frac{N}{N-N_r+1}\right\rfloor$ denotes the minimum number of servers needed to compute $f(W_{\mathcal{K}})$, given as in (\ref{multilinear_function}), where each of these servers computes a disjoint product of $M$ subfunctions, and 
\item the variable $\Delta_N=N-N^*\cdot (N-N_r+1)$ represents whether an additional server is needed to aid the computation, and if $\Delta_N\geq 1$, then $\Deltaprod$ denotes the number of subfunctions to be computed by the additional server, and similarly as above, $\mathbb{P}\Big(\prod\nolimits_{k\in \mathcal{S}\,:\, |\mathcal{S}|={\Deltaprod}} W_k\Big)=\epsilon_{\Deltaprod}$.
\end{itemize}
\end{prop}

\begin{proof}
See Appendix \ref{proof-prop:multi_shot_multilinear_function}.
\end{proof}

We will next detail two numerical examples (Subsections~\ref{ex:binary_lin_sep}-\ref{ex_prop:multi_shot_multilinear_function}) to showcase the achievable gains in the total communication cost for Proposition~\ref{prop:general_placement__correlated_Boolean_function} and Proposition~\ref{prop:multi_shot_multilinear_function}, respectively.

\section{Numerical Evaluations to Demonstrate the Achievable Gains}
\label{sec:examples}
Given $\mathcal{T}(N, K, K_c, M, N_r)$, 
to gain insight into our analytical results and demonstrate the savings in the total communication cost, we provide some numerical examples. To demonstrate  Proposition~\ref{prop:general_placement__correlated_Boolean_function}, in Subsection~\ref{ex:binary_lin_sep}, we focus on computing linearly separable functions, and in Subsection~\ref{ex_prop:multi_shot_multilinear_function} (cf.~Proposition~\ref{prop:multi_shot_multilinear_function}), we focus on multi-linear functions, respectively.

To that end, to characterize the performance of our characteristic graph-based approach for linearly separable functions, we denote by $\eta_{lin}$ the gain of the sum-rate for the characteristic graph-based approach given in (\ref{eq:bernouilli_correlated_subfunctions_communication_rate_upper_bound}) over the sum-rate of the distributed scheme of Wan {\em et al.} in \cite{wan2021distributed}, given in (\ref{eq:linearly_separable_communication_rate_upper_bound}), and by $\eta_{SW}$ the gain of the sum-rate in 
(\ref{eq:bernouilli_correlated_subfunctions_communication_rate_upper_bound}) over the sum-rate of the fully distributed approach of Slepian-Wolf \cite{SlepWolf1973}.
To capture general statistics, i.e., dataset skewness and correlations, and make a fair comparison, we adapt the transmission model of Wan {\em et al.} in \cite{wan2021distributed} via modifying the i.i.d. dataset assumption. 

We next study an example scenario (Subsection~\ref{ex:binary_lin_sep}) for computing a class of linearly separable functions (\ref{linearly_separable_functions}) over $\mathbb{F}_2$, where each of the demanded functions takes the form $f_j(W_{\mathcal{K}})=\sum\nolimits_{k\in [K]} \gamma_{jk} W_k \mod 2$, $j\in[K_c]$, under a specific correlation model across subfunctions. 
More specifically, when the subfunctions $W_k\sim{\rm Bern}(\epsilon)$ are identically distributed and correlated across $k\in [K]$, and $\Delta\in\mathbb{Z}^+$, we model the correlation across datasets (a) exploiting the joint PMF model in \cite[Theorem~1]{DinizBayesian2010} and (b) for a joint PMF described in Table~\ref{tab:joint_pmf}. 
Furthermore, we assume for $K_c>1$ that ${\bf \Gamma}=\{\gamma_{jk}\}\in\mathbb{F}_2^{K_c\times K}$ is full rank. 
For the proposed setting, we next demonstrate the achievable gains $\eta_{lin}$ 
of our proposed technique versus $\epsilon$ for computing (\ref{linearly_separable_functions}), as a function of skew, $\epsilon$, and correlation, $\rho$, of datasets, $K_c\in[N_r]<K$, and other system parameters, and showcase the results via Figures~\ref{fig:kc1corr}, \ref{fig:K_c_2_0_corr}, \ref{fig:kc2_special_func}, and \ref{fig:gain_for_kcgeneral_0corrolation}.

\subsection{Example case: Distributed computing of linearly separable functions over $\mathbb{F}_2$}
\label{ex:binary_lin_sep}


We consider the computation of linearly separable functions given in  (\ref{linearly_separable_functions}) for general topologies, with general $N$, $K$, $M$, $N_r$, $K_c$, over $\mathbb{F}_2$, with an identical skew parameter $\epsilon\in [0,1]$ for each subfunction, where $W_k\sim$Bern($\epsilon$), $k\in[K]$, using cyclic placement as in (\ref{cyclic_Zi}), and incorporating the correlation between the subfunctions, with the correlation coefficient denoted by $\rho$. We consider three scenarios, as described next:

\paragraph{\bf Scenario~I. The number of demanded functions is $K_c=1$, where the subfunctions could be uncorrelated or correlated.} 
This scenario is similar to the setting in \cite{wan2021distributed} whereas different from \cite{wan2021distributed} which is valid over a field of characteristic $q>3$, we consider $\mathbb{F}_2$, and in the case of correlations, i.e., when $\rho>0$, we capture the correlations across the transmissions (evaluated from subfunctions of datasets) from distributed servers, as detailed earlier in Section~\ref{results}. 
We first assume that the subfunctions are not correlated, i.e., $\rho=0$, and evaluate $\eta_{lin}$ for $f(W_{\mathcal{K}})=\sum\nolimits_{k\in [K]} W_k \mod 2$. 
The parameter of $f(W_{\mathcal{K}})$, i.e., the probability that $f(W_{\mathcal{K}})$ takes the value $1$ can be computed using the recursive relation:
\begin{align}
\label{epsilon_l_recursive}
\mathbb{P}\Big(\sum\limits_{k\in \mathcal{S}\,:\, |\mathcal{S}|=l\leq K} W_k \mod 2 = 1\Big)&=\sum\limits_{k\in\mathcal{S}\,:\,\,  |\mathcal{S}|=l\leq K\ , \, k\, {\rm odd}}\mathbb{P}(B(K,\epsilon))=k)\nonumber\\
&=(1-\epsilon_{l-1})\cdot \epsilon+\epsilon_{l-1}\cdot (1-\epsilon)\nonumber\\
&\doteq\,\epsilon_l\ ,\quad 1<l\leq K \ ,
\end{align}
where $B(K,\epsilon)$ is the binomial PMF,  
and $\epsilon_l$ is the probability that the modulo $2$ sum of any $1<l\leq K$ subfunctions taking the value one, with $W_k\sim{\rm Bern}(\epsilon)$ being i.i.d. across $k\in\mathcal{S}$, with the convention $\epsilon_1=\epsilon$.

Given $N_r$, we denote by $N^*=\left\lfloor \frac{N}{N-N_r+1}\right\rfloor$ the minimum number of servers, corresponding to the subset $\mathcal{N}^*\subseteq\Omega$, needed to compute $f(W_{\mathcal{K}})$, where each server, with a cache size of $M$, computes a sum of $M$ subfunctions, where across these $N^*$ servers, the sets of subfunctions are disjoint. Hence, $\mathbb{P}\Big(\sum\nolimits_{k\in \mathcal{S}\,:\, |\mathcal{S}|=M} W_k\Big)=\epsilon_M$. Furthermore, the variable $\Delta_N=N-N^*\cdot (N-N_r+1)$ represents whether additional servers in addition to $N^*$ servers are needed to aid the computation, and if $\Delta_N\geq 1$, then $\Delta\cdot \Delta_N \doteq \Deltaprod$ denotes the number of subfunctions to be computed by the set of additional servers, namely $\mathcal{I}^*\in\Omega$, and similarly as above, $\mathbb{P}\Big(\sum\nolimits_{k\in \mathcal{S}\,:\, |\mathcal{S}|={\Deltaprod}} W_k\Big)=\epsilon_{\Deltaprod}$, which is obtained evaluating $\epsilon_l$ at $l=\Deltaprod$.

Adapting (\ref{eq:linearly_separable_communication_rate_upper_bound}) for $\mathbb{F}_2$, we obtain the total communication cost $R_{ach}(lin)$ for computing the linearly separable function $f(W_{\mathcal{K}})=\sum\nolimits_{k\in [K]} W_k \mod 2$ %
as 
\begin{align}
\label{cost_lin_scenario_I}
R_{ach}(lin)= \sum\limits_{i=1}^{N_r} h\Big(\sum\limits_{k\in\mathcal{Z}_i}W_k\Big)=N_r\cdot h(\epsilon_M)  \ .  
\end{align}

Using Proposition~\ref{prop:general_placement__correlated_Boolean_function} and (\ref{epsilon_l_recursive}), we derive the sum rate for distributed lossless computing of $f(W_{\mathcal{K}})$ as
\begin{align}
\label{sum_rate_multishot_affine}
\sum\limits_{i\in\Omega}R_i\leq N^* \cdot h(\epsilon_M)+1_{\Delta_N>0}\cdot h(\epsilon_{\Deltaprod})\ ,
\end{align}
where the indicator function $1_{\Delta_N>0}$ captures the rate contribution from the additional server, if any.  Using (\ref{sum_rate_multishot_affine}), the gain $\eta_{lin}$ over the linearly separable solution of \cite{wan2021distributed} is presented as:
\begin{align}
\label{eta_lin_scenario_I}
\eta_{lin}=\frac{N_r\cdot h(\epsilon_M)}{N^* \cdot h(\epsilon_M)+1_{\Delta_N>0}\cdot h(\epsilon_{\Deltaprod})}\ ,    
\end{align}
where $h(\epsilon_{\Deltaprod})$ represents the rate needed from the set of additional servers $\mathcal{I}^*\in\Omega$, aiding the computation through communicating the sum of the remaining subfunctions in the set $\mathcal{C}\subseteq \mathcal{Z}_{\mathcal{I}^*}$, where the summation for these remaining functions in $\mathcal{C}\subseteq \mathcal{Z}_{\mathcal{I}^*}$ is denoted as $\sum_{k\in\mathcal{C}\subseteq \mathcal{Z}_{\mathcal{I}^*}\ :\ k\notin \bigcup_{i\in \mathcal{N}^*}\mathcal{Z}_i \ , \,\,  |\mathcal{C}|=\Deltaprod} W_k$, that cannot be captured by the set $\mathcal{N}^*$.

Given $K_c=1$ for the given modulo $2$ sum function, we next incorporate the correlation model in \cite{DinizBayesian2010}, for each $W_k$, identically distributed with $W_k\sim{\rm Bern}(\epsilon)$, and correlation $\rho$ across any two subfunctions. The formulation in \cite{DinizBayesian2010} yields the following PMF for $f(W_{\mathcal{K}})$:
\begin{align}
\label{ex:corr_model_probability} 
\mathbb{P}(f(W_{K})=y)&= \binom{K}{y} \epsilon^{y}(1-\epsilon)^{K-y} (1-\rho) \cdot 1_{y\in A_1}\nonumber\\
&+  \epsilon^{\frac{y}{K}}(1-\epsilon)^{\frac{K-y}{K}}\rho \cdot 1_{y\in A_2} \ , \, y\in \{0,\cdots,K\} \ , 
\end{align}
where $1_{y\in A_1}$ and $1_{y\in A_2}$ are indicator functions, where ${A_1}=\{0,1,\cdots,K\}$ and ${A_2}=\{0,K\}$.

We depict the behavior of our gain, $\eta_{lin}$, using the same topology $\mathcal{T}(N, K, K_c, M, N_r)$ as in \cite{wan2021distributed}, with different system parameters $(N, K, M, N_r)$, under $\rho=0$ in Figure~\ref{fig:kc1corr}-(Left). As we increase both $N$ and $K$, along with the number of active servers, $N_r$, the gain, $\eta_{lin}$, of the characteristic graph approach increases. This stems from the characteristic graph approach, to compute functions $f(W_{\mathcal{K}})$ of $W_{\mathcal{K}}$ using $N^*$ servers.
From Figure~\ref{fig:kc1corr}-(Right), it is evident that by capturing correlations between the subfunctions, hence across the servers' caches, $\eta_{lin}$ grows more rapidly until it reaches the maximum of~(\ref{eta_lin_scenario_I}), corresponding to $\eta_{lin}=\frac{N_r}{N^*}=10$, attributed to full correlation (see, Figure~\ref{fig:kc1corr}-(Right)).

\begin{figure*}[t!]
\centering
\includegraphics[width=0.4\textwidth]{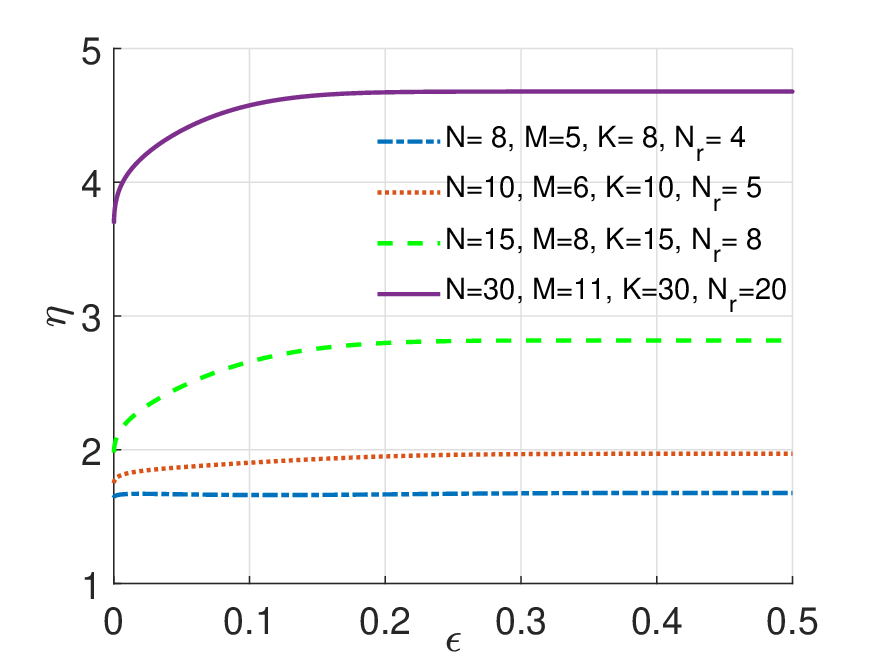}
\includegraphics[width=0.4\textwidth]{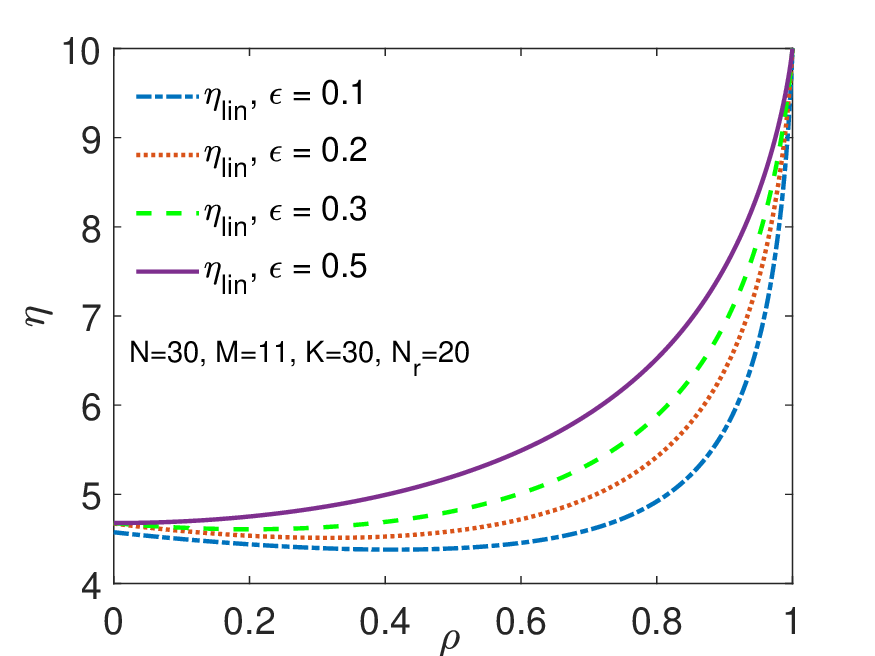}
\vspace{-0.2cm}
\caption{\small{The gain $\eta_{lin}$ of the characteristic graph approach for $K_c=1$, in 
Subsection~\ref{ex:binary_lin_sep} (Scenario~I). (Left) $\rho=0$ for various distributed topologies. (Right) The correlation model given as (\ref{ex:corr_model_probability}) for $\mathcal{T}(30, 30, 1, 11, 20)$ with different $\epsilon$ values.}}
\label{fig:kc1corr}
\vspace{-0.2cm}
\end{figure*}

What we can also see is that for $\rho=0$, the gain rises with the increase of $\epsilon$ and linearly grows with $\frac{N_r}{N^*}$. As $\rho$ increases reaching its maximum at $\rho=1$, the gain is maximized, yielding the minimum communication cost that can be achieved with our technique. Here, the gain $\eta_{lin}$ is dictated by the topology and is given as $\eta_{lin}=\frac{N_r}{N^*}$. This linear relation shows that this specific topology can provide a very substantial reduction in the total communication cost, as $\rho$ goes to $1$, over the state-of-the-art \cite{wan2021distributed}, as shown in Figure~\ref{fig:kc1corr}-(Right) via the purple (solid) curve. Furthermore, one can draw a comparison between the characteristic graph approach and the approach in \cite{SlepWolf1973}. Here, we represent the gain as $\eta_{SW}$. It is noteworthy that the sum rate of all servers using the coding approach of Slepian-Wolf \cite{SlepWolf1973} is $R_{ach}(SW)=H(W_{\mathcal{K}})$. With $\rho=0$, this expression simplifies to $R_{ach}(SW)=K\cdot H(W_k)$, resulting again in a substantial reduction in the communication cost as we see from $R_{ach}(lin)$ in (\ref{cost_lin_scenario_I}) for the same topology of the purple (solid) curve as shown in Figure~\ref{fig:kc1corr}-(Right).

\paragraph{\bf Scenario~II. The number of demanded functions is $K_c=2$, where the subfunctions could be uncorrelated or correlated.} 
To gain insights into the behavior of $\eta_{lin}$, we consider an example distributed computation model with $K=N=3$, $N_r=2$, where the subfunctions $W_1,\,W_2,\,W_3$ are assigned to $X_1$, $X_2$ and $X_3$ in a cyclic manner, with $h(W_k)= \epsilon$, $k\in[3]$, and $K_c=2$ with  $f_1(W_{\mathcal{K}})=W_2$, and $f_2(W_{\mathcal{K}})=W_2+ W_3$.

Given $N_r=2$, using the characteristic graph approach for individual servers, an achievable compression scheme, for a given ordering $i$ and $j$ of server transmissions, relies on first compressing of the characteristic graph $G_{X_i}$ constructed by server $i\in\Omega$ that has no side information, and then the conditional rate needed for compressing the colors of  $G_{X_j}$ for any other server $j\in \Omega\backslash i$ via incorporating the side information $Z_i=g_i(X_i)$ obtained from server $i\in\Omega$. Thus, contrasting the total communication cost associated with the possible orderings, the minimum total communication cost $R_{ach}(G)$ can be determined\footnote{We can generalize (\ref{min_charac_entorpy_over_servers}) to $N_r>2$, where, for a given ordering of server transmissions, any consecutive server that transmits sees all previous transmissions as side information and the best ordering that has the minimum total communication cost, i.e., $R_{ach}(G)$.}. 
The achievable sum rate here takes the form 
\begin{align}
\label{min_charac_entorpy_over_servers}
 R_{ach}(G)= \min\{H_{G_{X_1}}(X_1)+ H_{G_{X_2}}(X_2\,\vert\,Z_1),\quad  H_{G_{X_2}}(X_2)+ H_{G_{X_1}}(X_1\,\vert\,Z_2)\} \ . 
\end{align}

\begin{figure}
    \centering
    \includegraphics[width=0.8\linewidth]{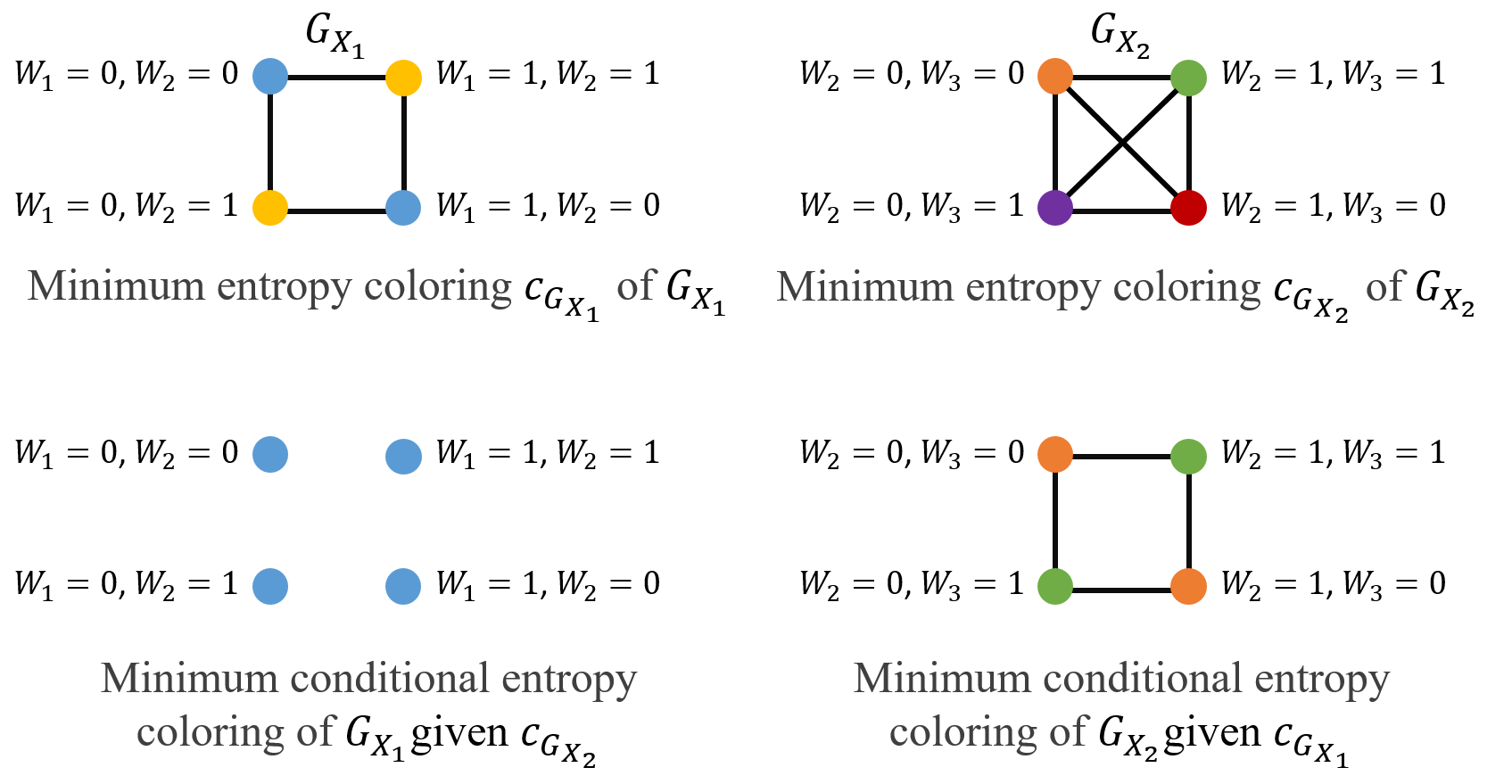}
    \caption{\small{Colorings of graphs in Subsection~\ref{ex:binary_lin_sep} (Scenario~II). (Top Left-Right) Characteristic graphs $G_{X_1}$ and $G_{X_2}$, respectively. 
    (Bottom Left-Right) The minimum conditional entropy colorings of $G_{X_1}$ given $c_{G_{X_2}}$, and $G_{X_2}$ given $c_{G_{X_1}}$, respectively.}}
    \label{fig:chargraph}
\end{figure}

Focusing on the characteristic graph approach, we illustrate how each server builds its union characteristic graph for simultaneously computing $f_1$ and $f_2$ according to (\ref{union_graph}) (as detailed in Appendix~\ref{app-graph-encode}), in Figure~\ref{fig:chargraph}. 
In (\ref{min_charac_entorpy_over_servers}), the first term corresponds to $G_{X_1}=(V_{X_1},E_{X_1})$ where $V_{X_1}=\{0,1\}^2$ is  built using the support of $W_1$ and $W_2$, and the edges $E_{X_1}$ are built based on the rule that $(x_1^1,x_1^2)\in E_{X_1}$ if $F(x_1^1,x_2)\neq F(x_1^2,x_2)$ for some $x_2\in V_{X_2}$, where, as we see here, requires $2$ colors. Similarly, server $2$ constructs $G_{X_2}=(V_{X_2},E_{X_2})$ given $Z_1$, where $V_{X_2}=\{0,1\}^2$ using the support of $W_2$ and $W_3$, and where $Z_1$ determines $f_1=W_2$, and hence, to compute $f_2=W_2+W_3$ given $f_1=W_2$, any two vertices taking values\footnote{Here, $x_2^1=(w_2^1,w_3^1)$ and $x_2^2=(w_2^2,w_3^2)$ represent two different realizations of the pair of subfunctions $W_2$ and $W_3$.} $x_2^1=(w_2^1,w_3^1)\in V_{X_2}$ and $x_2^2=(w_2^2,w_3^2)\in V_{X_2}$ are connected if $w_3^1\neq w_3^2$. Hence, we require $2$ distinct colors for $G_{X_2}$. 
As a result, the first term yields a sum rate of $h(\epsilon)+h(\epsilon)=2h(\epsilon)$.
Similarly, the second term of (\ref{min_charac_entorpy_over_servers}) captures the impact of $G_{X_2}=(V_{X_2}, E_{X_2})$, where server $2$ builds $G_{X_2}$ using the support of $W_2$ and $W_3$, and $G_{X_2}$ is a complete graph to distinguish all possible binary pairs to compute $f_1$ and $f_2$, requiring $4$ different colors. Given $Z_2$, both $f_1$ and $f_2$ are deterministic. Hence, given $Z_2$, $G_{X_1}$ has no edges, which means that $H_{G_{X_1}}(X_1\,\vert\, Z_2)=0$. 
As a result, the ordering of server transmission given by the second term of (\ref{min_charac_entorpy_over_servers}) yields the same sum rate of $2h(\epsilon)+0=2h(\epsilon)$. 
For this setting, the minimum required rate is $R_{ach}(G) = 2h(\epsilon)$, and the configuration captured by the second term provides a lower recovery threshold of $N_r=1$ versus $N_r=2$ for the configurations of server transmissions given by the first term  (\ref{min_charac_entorpy_over_servers}). The different $N_r$ achieved by these two configurations is also captured by Figure~\ref{fig:chargraph}.  

Alternatively, in the linearly separable approach \cite{wan2021distributed}, $N_r$ servers transmit the requested function of the datasets stored in their caches. For distributed computing of $f_1$ and $f_2$, servers $1$ and $2$ transmit at rate $H(W_2)=h(\epsilon)$, for computing $f_1$, and at rate $H(W_2 + W_3)$, for function $f_2$. As a result, the achievable communication cost is given by $R_{ach}(lin) = h(\epsilon) + h(W_2 + W_3)$. 
Here, for a fair comparison, we update the model studied in \cite{wan2021distributed} to capture the correlation within each server without accounting for the correlation across the servers.

Under this setting, for $\rho=0$, we see that the gain $\eta_{lin}$ of the characteristic graph approach over the linearly separable solution of \cite{wan2021distributed} for computing $f_1$ and $f_2$ as a function of $\epsilon \in [0, 1]$ takes the form
\begin{align}
\label{eta_scenario_II_rho_0}
\eta_{lin}(\epsilon) &= \frac{h(\epsilon) + h(2\epsilon(1-\epsilon))}{2h(\epsilon)}
\begin{cases}
    =1 & ,\,\, \epsilon = \{\frac{1}{2}\} \ , \\
    >1 & ,\,\, \epsilon\in  [0,1]\backslash \{\frac{1}{2}\} \ , 
\end{cases}
\end{align}
where $\eta_{lin}(\epsilon)>1$ for $\epsilon\neq \frac{1}{2}$ follows from the concavity of $h(\cdot)$ that yields the inequality $h(2\epsilon(1-\epsilon)) \geq h(\epsilon)$. Furthermore, $\eta_{lin}$ approaches $1.5$ as $\epsilon \to \{0, 1\}$ (see Figure~\ref{fig:K_c_2_0_corr}).
\begin{figure}
    \centering
    \includegraphics[width=0.45\linewidth]{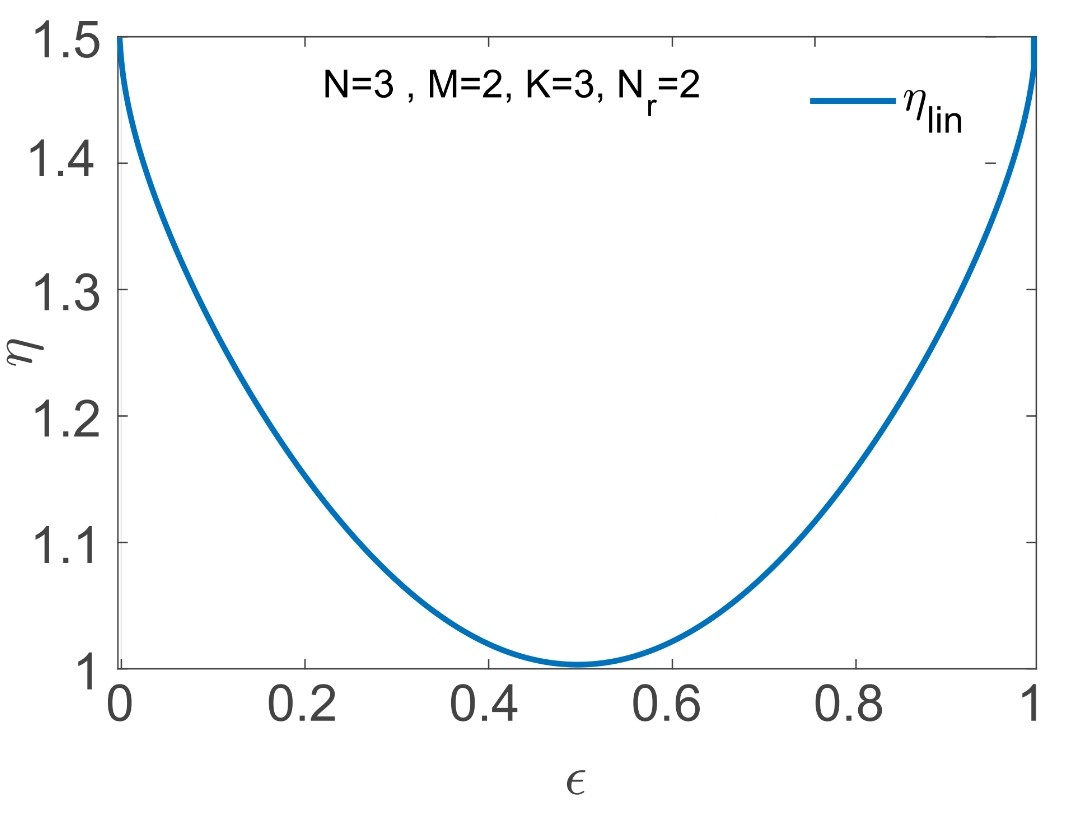}
    \vspace{-0.3cm}
    \caption{$\eta_{lin}$ in (\ref{eta_scenario_II_rho_0}) versus $\epsilon$, for distributed computing of $f_1=W_2$ and $f_2=W_2+ W_3$, where $K_c=2$, $N_r=2$, with $\rho=0$, 
    in Subsection~\ref{ex:binary_lin_sep} (Scenario~II).}
    \label{fig:K_c_2_0_corr}
\end{figure}

\begin{table}[h!]\small
\centering
\begin{tabular}{cc|cc}
   &  $P_{W_2,W_3}(W_2,W_3)$     & $W_2=0$  & $W_2=1$  \\ \hline
   & $W_3=0$ & $(1-\epsilon)(1-p')$ & $\epsilon p$ \\
   & $W_3=1$ & $\epsilon p$ & $\epsilon(1-p)$ \\
\end{tabular}
\caption{\small{Joint PMF $P_{W_2,W_3}$ of $W_2$ and $W_3$ with a crossover parameter $p$, 
in Subsection~\ref{ex:binary_lin_sep} (Scenario~II).}}
\label{tab:joint_pmf}
\end{table}

We next examine the setting where the correlation coefficient $\rho$ is nonzero, using the joint PMF $P_{W_2, W_3}$, as depicted in Table~\ref{tab:joint_pmf}, of the required subfunctions ($W_2$ and $W_3$) in computing $f_1$ and $f_2$. This PMF describes the joint PMF corresponding to a binary non-symmetric channel model, where the correlation coefficient between $W_2$ and $W_3$ is $\rho=\frac{1-p}{1-\epsilon}$, and where
$p'=\frac{\epsilon p}{1-\epsilon}$. Thus our gain here compared to the linearly separable encoding and decoding approach of \cite{wan2021distributed} is given as
\begin{align}
\label{eta_scenario_II_rho_nonzero_TableII}
\eta_{lin}=\frac{H(W_2)+H(W_2 + W_3)}{H(W_2, W_3)}
=\frac{h(\epsilon)+h(2\epsilon p)}{h(\epsilon)+(1-\epsilon)h\big(\frac{\epsilon p}{1-\epsilon}\big)+\epsilon h(p)} \ .
\end{align} 

We consider now the correlation model in Table~\ref{tab:joint_pmf}, where coefficient $\rho$ rises in $\epsilon$ for a fixed $p$. In Figure~\ref{fig:kc2_special_func}-(Left), we illustrate the behavior of $\eta_{lin}$, given by (\ref{eta_scenario_II_rho_nonzero_TableII}), for computing $f_1$ and $f_2$ for $N_r=2$ as a function of $p$ and $\epsilon$, where for this setting, the correlation coefficient $\rho$ is a decreasing function of $p$ and an increasing function of $\epsilon$. 
We observe from (\ref{eta_scenario_II_rho_nonzero_TableII}) that the gain $\eta_{lin}$ satisfies $\eta_{lin}\geq 1$ for all $\epsilon\in[0,1]$, which monotonically increases in $p$ --- and hence monotonically decreases in $\rho$ due to the relation $\rho=\frac{1-p}{1-\epsilon}$ --- as a function of the deviation of $\epsilon$ from $1/2$. 
For $\epsilon\in (0.5, 1]$, $\eta_{lin}$ increases in $\epsilon$. For example, for $p=0.1$ then $\eta_{lin}(1)=1.28$, as depicted by the green (solid) curve. 
Similarly, given $\epsilon\in [0, 0.5)$, decreasing $\epsilon$ results in $\eta_{lin}$ to exhibit a rising trend, e.g., for $p=0.9$ then $\eta_{lin}(0)=1.36$, as shown by the red (dash-dotted) curve. 
As $p$ approaches one, $\eta_{lin}$ goes to $1.5$ as $\epsilon$ tends to zero, which can be derived from (\ref{eta_scenario_II_rho_nonzero_TableII}). 
We here note that the gains are generally smaller than in the previous set of comparisons as shown in Figure~\ref{fig:K_c_2_0_corr}.

More generally, given a user request, consisting of $K_c=2$ linearly separable functions (i.e., satisfying  (\ref{linearly_separable_functions})), and after considering (\ref{eta_scenario_II_rho_nonzero_TableII}) beyond $N_r=2$, we see that $\eta_{lin}$ is at most $N_r$ as $\rho$ approaches one. 
We next use the joint PMF model used in obtaining (\ref{ex:corr_model_probability}), where we observe that $f_2\sim ((1-\epsilon)^2(1-\rho)+(1-\epsilon)\rho,\, 2\epsilon(1-\epsilon)(1-\rho),\, \epsilon^2(1-\rho)+\epsilon\rho)$, to see that the gain takes the form 
\begin{align}
\label{eta_scenario_II_rho_nonzero_Eqn3}
\eta_{lin}
=\frac{h(\epsilon)+H(f_2)}
{h(\epsilon) + (1-\epsilon)h(\zeta_1) + \epsilon h(\zeta_2)} \ ,
\end{align}
where $\zeta_1=(1-\epsilon)(1-\rho)+\rho$, and $\zeta_2=(1-\epsilon)(1-\rho)$. For this model, we illustrate $\eta_{lin}$ versus $\epsilon$ in Figure~\ref{fig:kc2_special_func}-(Right) for different $\rho$ values. 
Evaluating (\ref{eta_scenario_II_rho_nonzero_Eqn3}), the peak achievable gain is attained when $\rho=1$ at $f_2\sim ((1-\epsilon),\, 0,\, \epsilon)$, yielding $H(W_2+W_3)=h(\epsilon)$ and $H(W_3\,\vert\,W_2)=(1-\epsilon)h(\rho)=0$, and hence, a gain $\eta_{lin} = N_r=2$, as shown by the purple (solid) curve.    
On the other hand, for $\rho=0$, we observe that $f_2\sim ((1-\epsilon)^2,\, 2\epsilon(1-\epsilon),\, \epsilon^2)$, yielding $H(W_2+W_3)=h((1-\epsilon)^2,\, 2\epsilon(1-\epsilon),\, \epsilon^2)=h(2\epsilon(1-\epsilon))+((1-\epsilon)^2+\epsilon^2)h\Big(\frac{\epsilon^2}{\epsilon^2+(1-\epsilon)^2}\Big)$ and $H(W_3\,\vert\,W_2)=(1-\epsilon)h(\epsilon) + \epsilon h(\epsilon)=h(\epsilon)$, and hence, it can be shown that the gain is lower bounded as $\eta_{lin}\geq 1.25$.

\begin{figure*}[t!]
    \centering   
    \includegraphics[width=0.45\textwidth]{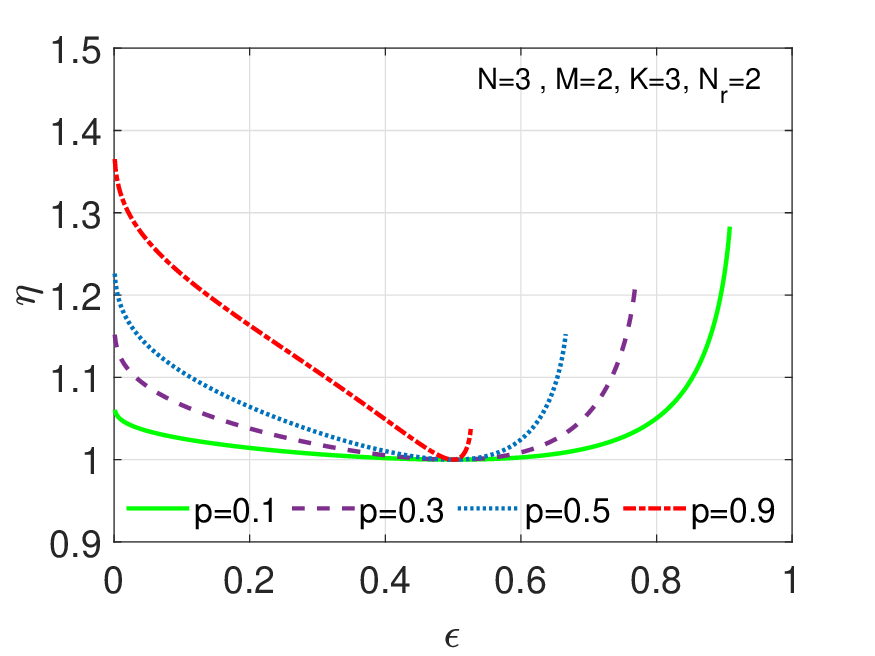}
    \includegraphics[width=0.45\textwidth]{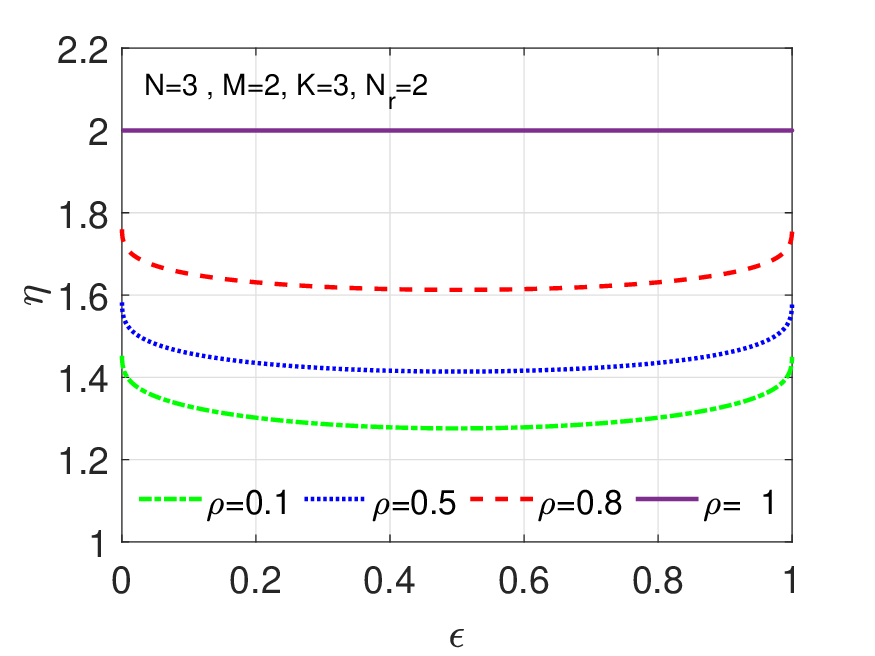}
    \vspace{-0.2cm}
    \caption{\small{$\eta_{lin}$ versus $\epsilon$, for distributed computing of $f_1=W_2$ and $f_2=W_2+ W_3$, where $K_c=2$, $N_r=2$, 
    in Subsection~\ref{ex:binary_lin_sep}, using different joint PMF models for $P_{W_2,W_3}$ (Scenario~II). 
    (Left) $\eta_{lin}$ in (\ref{eta_scenario_II_rho_nonzero_TableII}) for the joint PMF in Table~\ref{tab:joint_pmf}, for different values of $p$. 
    (Right) $\eta_{lin}$ for the joint PMF in (\ref{ex:corr_model_probability}), for different values of $\rho$ .}}
\label{fig:kc2_special_func}
\end{figure*}

\paragraph{\bf Scenario~III. The number of demanded functions is $K_c\in[N_r]$, and the number of datasets is equal to the number of servers, i.e., $K=N$, where the subfunctions are uncorrelated.} 
We now provide an achievable rate comparison between the approach in \cite{wan2021distributed}, and our graph-based approach, as summarized by our Proposition~\ref{prop_cyclic_placement_iid_uniform_source-linear_function_graph-rate_UB} that generalizes the result in \cite[Theorem 2]{wan2021distributed} to finite fields with characteristics $q\geq 2$, for the case of $\rho=0$. 

Here, to capture dataset skewness and make a fair comparison, we adapt the transmission model of Wan {\em et al.} in \cite{wan2021distributed} via modifying the i.i.d. dataset assumption and taking into account the skewness incurred within each server in determining the local computations $\sum\nolimits_{k\in \mathcal{S}\,:\, |\mathcal{S}|=M} W_k$ at each server.

For the linearly separable model in (\ref{linearly_separable_functions}) adapted to account for our setting, exploiting the summation $\sum_{k\in \mathcal{Z}_i}W_k$, and $\epsilon_M$ given in (\ref{sum_rate_multishot_affine}), the communication cost for a general number of $K_c$ with $\rho=0$ is expressed as:
\begin{align}
\label{ex:general_K_c_with_rho_0_lin}        
R_{ach}(lin)= N_r\cdot h(\epsilon_M) \ .
\end{align}
In (\ref{ex:general_K_c_with_rho_0_lin}), as $\epsilon$ approaches $0$ or $1$, then  $h(\epsilon_M)\to 0$. Subsequently, the achievable communication cost for the characteristic graph model can be determined as
\begin{align}
\label{ex:general_K_c_with_rho_0_graph}
R_{ach}(G)= K_c\cdot {N^*}\cdot h(\epsilon) \ . 
\end{align}
To understand the behavior of $\eta_{lin}=\frac{N_r}{K_c {N^*}}\cdot \frac{h(\epsilon_M)}{h(\epsilon)}$, knowing that $\frac{N_r}{K_c N^*}$ is a fixed parameter, we need to examine the dynamic component $\frac{h(\epsilon_M)}{h(\epsilon)}$. 
Exploiting Schur concavity\footnote{A real-valued function $f:\mathbb{R}^n \to\mathbb{R}$ is Schur concave if $ f(x_1, x_2, \cdots, x_n) \leq f(y_1, y_2, \cdots, y_n)$ holds whenever $ (x_1, x_2, \cdots, x_n) $ majorizes $(y_1, y_2, \cdots, y_n)$, i.e., $\sum_{i=1}^{k} x_i\geq \sum_{i=1}^{k} y_i$, for all $k\in [n]$ \cite{boland1990some}.} for the binary entropy function, which tells us that $h(\mathbb{E}[X])\geq \mathbb{E}[h(X)]$, 
we can see that as $\epsilon$ approaches $0$ or $1$, then 
\begin{align}
\label{gain_scaling_UB_scenarioIII}
\lim_{\epsilon\to \{0,\,1\} }\frac{h(\epsilon_M)}{h(\epsilon)}\leq M \ , \quad M\in\mathbb{Z}^+\ ,
\end{align}
where the inequality between the left and right-hand sides becomes loose as a function of $M$. 
As a result, as $\epsilon$ approaches $0$ or $1$, then $\eta_{lin}\approx M\cdot\frac{N_r}{K_c \cdot {N^*}}$, which follows from exploiting  (\ref{ex:general_K_c_with_rho_0_lin}), (\ref{ex:general_K_c_with_rho_0_graph}) and the achievability of the upper bound in (\ref{gain_scaling_UB_scenarioIII}). We illustrate the upper bound on $\eta_{lin}$ in Figure~\ref{fig:gain_for_kcgeneral_0corrolation}, and demonstrate the $\eta_{lin}$ behavior for $K_c$ demanded functions across various topologies with circular dataset placement, namely for various $K=N$, i.e., when the amount of circular shift between two consecutive servers is $\Delta=\frac{K}{N}=1$ and the cache size is $M=N-N_r+1$, and for $\rho=0$ and $\epsilon \leq 1/2$. We focus only on plotting $\eta_{lin}$ for $\epsilon \leq 1/2$, accounting for the symmetry of the entropy function. 
Therefore, we only plot for $\epsilon\leq 1/2$. The multiplicative coefficient $\frac{N_r}{K_c N^*}$ of $\eta_{lin}$ determines the growth, which is depicted by the curves. 

\begin{figure}[t!]
\centering
\includegraphics[width=0.48\linewidth]{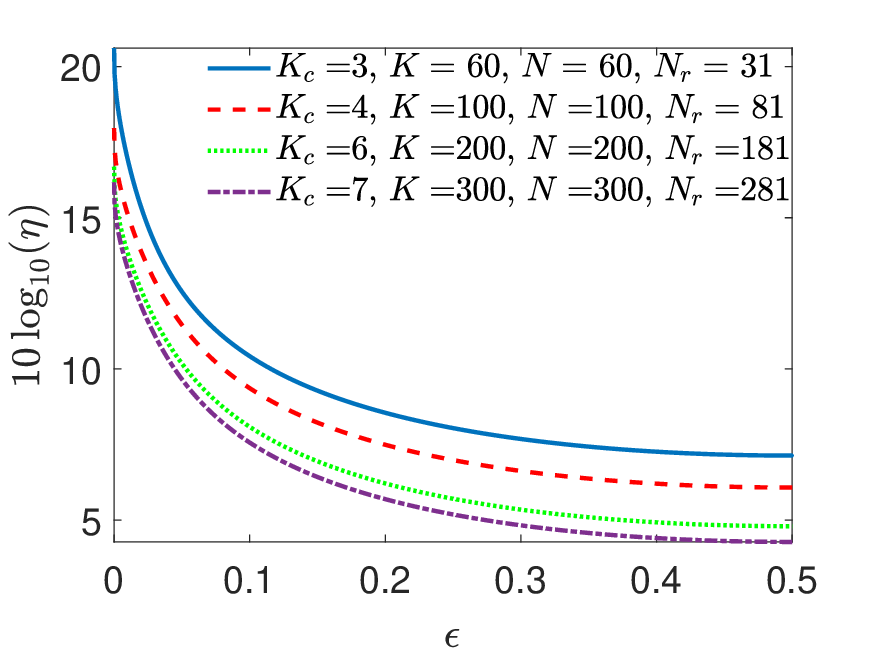}
\vspace{-0.3cm}
\caption{\small{$\eta_{lin}$ in a logarithmic scale versus $\epsilon$ for $K_c$ demanded functions for various values of $K_c$, with $\rho=0$ for different topologies, as detailed in Subsection~\ref{ex:binary_lin_sep} (Scenario~III).}}
\label{fig:gain_for_kcgeneral_0corrolation}
\end{figure}

Thus we see that for a given topology $\mathcal{T}(N, K, K_c, M, N_r)$ with $K_c$ demanded functions, for $\rho=0$, using (\ref{gain_scaling_UB_scenarioIII}), we see that $\eta_{lin}$ exponentially grows with term $1-\epsilon$ for $\epsilon\in[0,1/2]$\footnote{Here we note that the behavior of $\eta_{lin}$ is symmetric around $\epsilon=1/2$.}, and very substantial reduction in the total communication cost is possible as $\epsilon$ approaches $\{0,1\}$ as shown in Figure~\ref{fig:gain_for_kcgeneral_0corrolation} by the blue (solid) curve. The gain over \cite[Theorem 2]{wan2021distributed}, $\eta_{lin}$, for a given topology, changes proportionally to $\frac{N_r}{K_c N^*}$. The gain over \cite{SlepWolf1973}, $\eta_{SW}$, for $\rho=0$ linearly scales\footnote{Incorporating the dataset skew to Proposition~\ref{prop_cyclic_placement_iid_uniform_source-linear_function_graph-rate_UB} (\cite[Theorem 2]{wan2021distributed}), $R_{ach}(lin)$ is simplified to (\ref{ex:general_K_c_with_rho_0_lin}), which from (\ref{gain_scaling_UB_scenarioIII}) can linearly grow in $M=N-N_r+1$ at high skew, explaining the inferior performance of Proposition~\ref{prop_cyclic_placement_iid_uniform_source-linear_function_graph-rate_UB} over \cite{SlepWolf1973} as a function of the skew.} with $\frac{K}{K_c N^*}$. 
For instance, the gain for the blue (solid) curve in Figure~\ref{fig:gain_for_kcgeneral_0corrolation} is $\eta_{SW}=10$.

In general, other functions in $\mathbb{F}_2$, such as bitwise $AND$ and the multi-linear function (see e.g., Proposition~\ref{prop:multi_shot_multilinear_function}), are more skewed and have lower entropies than linearly separable functions, and hence are easier to compute. Therefore, the cost given in (\ref{ex:general_K_c_with_rho_0_graph}) can serve as an upper bound for the communication costs of those more skewed functions in $\mathbb{F}_2$.

We have here provided insights into the achievable gains in communication cost for several scenarios.   
We leave the study of $\eta_{lin}$ for more general topologies $\mathcal{T}(N, K, K_c, M, N_r)$ and correlation models beyond (\ref{ex:corr_model_probability}) devised for linearly separable functions, and beyond the joint PMF model in Table~\ref{tab:joint_pmf}, as future work.

Proposition~\ref{prop:multi_shot_multilinear_function} illustrates the power of the characteristic graph approach in decreasing the communication cost for distributed computing of multi-linear functions, given as in (\ref{multilinear_function}), compared to recovering the local computations $\prod_{k\in \mathcal{S}\,:\, |\mathcal{S}|=M} W_k$ using \cite{SlepWolf1973}. 
We denote by $\eta_{SW}$ the gain of the sum-rate for the graph entropy-based approach given in (\ref{sum_rate_multishot_multilinear}) --- using the conditional entropy-based sum-rate expression in (\ref{sum_rate_cond_graph_entropies}) --- over the sum-rate of the fully distributed scheme of Slepian-Wolf \cite{SlepWolf1973}  
for computing (\ref{multilinear_function}). 
For the proposed setting, we next showcase the achievable gains $\eta_{SW}$ of Proposition~\ref{prop:multi_shot_multilinear_function} via an example and showcase the results via Figure~\ref{fig:gains_computing_bilinearfunctions}.

\subsection{Distributed computation of $K$-multi-linear functions over $\mathbb{F}_2$}
\label{ex_prop:multi_shot_multilinear_function}

We study the behaviors of $\eta_{SW}$ 
versus the skewness parameter $\epsilon$ for computing the multi-linear function given in  (\ref{multilinear_function}) for i.i.d. uniform $W_k\sim{\rm Bern}(\epsilon)$, $\epsilon\in[0,1/2]$ across $k\in[K]$, and for a given $\mathcal{T}(N, K, K_c, M, N_r)$ with parameters $N$, $K$, $M=\Delta(N-N_r+1)$, such that $N_r=N-1$, $K_c=1$, $\rho=0$, and the number of replicates per dataset is $\frac{MN}{K}=2$. We use Proposition~\ref{prop:multi_shot_multilinear_function} to determine the sum-rate upper bound, and illustrate the gains $10\log_{10}(\eta_{SW})$ in decibels versus $\epsilon$ in Figure~\ref{fig:gains_computing_bilinearfunctions}. 
\begin{figure*}[t!]
\centering
\includegraphics[width=0.48\textwidth]{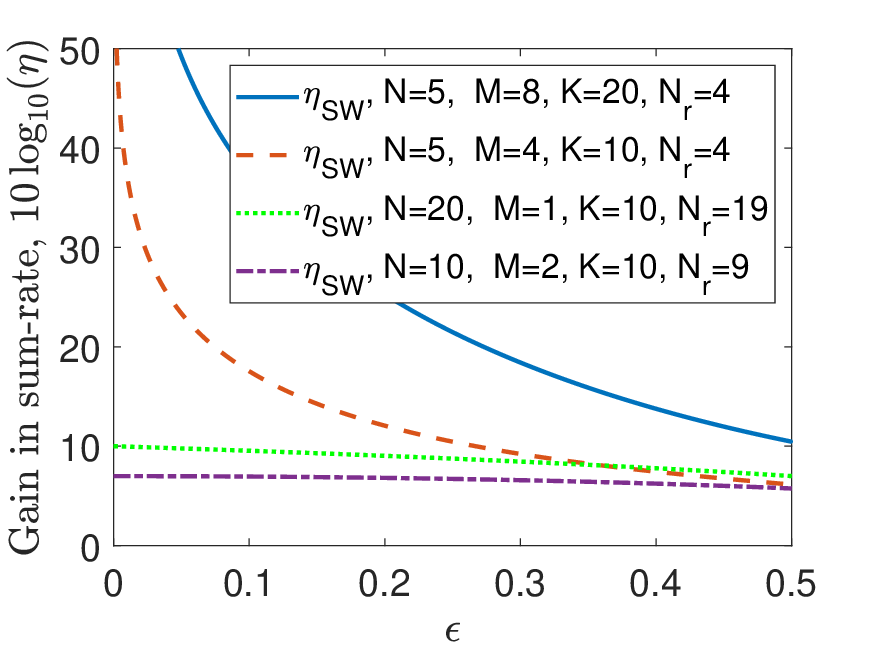}
\includegraphics[width=0.48\textwidth]{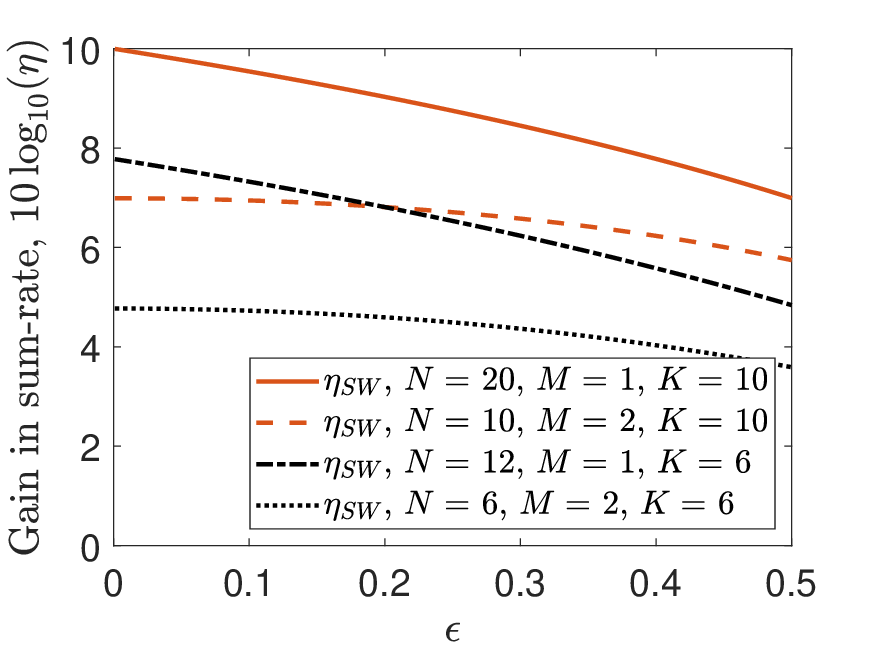}
\vspace{-0.3cm}
\caption{\small{Gain $10\log_{10}(\eta_{SW})$  versus $\epsilon$ for computing (\ref{multilinear_function}), where $K_c=1$, $\rho=0$, $N_r=N-1$. (Left) The set of parameters $N$, $K$, and $M$ are indicated for each configuration. (Right) $10\log_{10}(\eta_{SW})$ versus $\epsilon$ to observe the effect of $N$ for a fixed total cache size $MN$ and fixed $K$.}}\label{fig:gains_computing_bilinearfunctions}
\end{figure*}

From the numerical results in Figure~\ref{fig:gains_computing_bilinearfunctions} (Left), we observe that the sum-rate gain of the graph entropy-based approach versus the fully distributed approach of \cite{SlepWolf1973}, $\eta_{SW}$, could reach up to more than $10$-fold gain in compression rate for uniform and up to $10^6$-fold for skewed data. 
The results on $\eta_{SW}$ showcase that our proposed scheme can guarantee an exponential rate reduction over \cite{SlepWolf1973} as a function of decreasing $\epsilon$. Furthermore, the sum-rate gains scale linearly with the cache size $M$, which scales with $K$ given $N_r=N-1$. 
Note that $\eta_{SW}$ diminishes with increasing $N$ when $M$ and $\Delta$ are kept fixed.
In Figure~\ref{fig:gains_computing_bilinearfunctions} (Right), for $M\ll K$, a fixed total cache size $MN$, and hence fixed $K$, the gain $\eta_{SW}$ for large $N$ and small $M$ is higher versus small $N$ and large $M$, demonstrating the power of the graph-based approach as the topology becomes more and more distributed. 
%

\section{Conclusion}
\label{sec:conclusion}

In this paper, we devised a distributed computation framework for general function classes in multi-server, multi-function, single-user topologies. 
Specifically, we analyzed the upper bounds for the communication cost for computing in such topologies, exploiting K\"orner's characteristic graph entropy, by incorporating the structures in the dataset and functions, as well as the dataset correlations. 
To showcase the achievable gains of our framework, and perceive the roles of dataset statistics, correlations, and function classes, we performed several experiments under cyclic dataset placement over a field of characteristic two. 
Our numerical evaluations for distributed computing of linearly separable functions, as demonstrated 
in Subsection~\ref{ex:binary_lin_sep} via three scenarios, indicate that by incorporating dataset correlations and skew, it is possible to achieve a very substantial reduction 
in the total communication cost 
over the state-of-the-art.
Similarly, for distributed computing of multi-linear functions, 
in Subsection~\ref{ex_prop:multi_shot_multilinear_function}, we demonstrate 
a very substantial reduction in the total communication cost 
versus the state-of-the-art. 
Our main results (Theorem~\ref{theo_cyclic_placement_general_source-general_function_graph-rate_UB} and Propositions~\ref{prop_cyclic_placement_iid_uniform_source-linear_function_graph-rate_UB}, \ref{prop:general_placement__correlated_Boolean_function}, and \ref{prop:multi_shot_multilinear_function}) and observations through the examples help us gain insights into reducing the communication cost of distributed computation by taking into account the structures of datasets (skew and correlations) and functions (characteristic graphs).

Potential directions include providing a tighter achievability result for Theorem~\ref{theo_cyclic_placement_general_source-general_function_graph-rate_UB} and devising a converse bound on the sum-rate. They involve conducting experiments under the scheme of the coded scheme of Maddah-Ali and Niesen detailed in \cite{MaddahAli2013Journal} in order to capture the finer-grained granularity of placement that can help tighten the achievable rates. 
They also involve, beyond the special cases detailed in Propositions~\ref{prop_cyclic_placement_iid_uniform_source-linear_function_graph-rate_UB}, \ref{prop:general_placement__correlated_Boolean_function}, and \ref{prop:multi_shot_multilinear_function}, exploring the achievable gains for a broader set of distributed computation scenarios, e.g., over-the-air computing, cluster computing, coded computing, distributed gradient descent, or more generally, distributed optimization and learning, and goal-oriented and semantic communication frameworks, that can be reinforced by compression by capturing the skewness, correlations, and placement of datasets, the structures of functions, and topology.



\section*{Acknowledgments}
The authors thank Prof. Kai Wan at the Huazhong University of Science and Technology, Wuhan, China for interesting discussions.

\appendix



\section{Technical Preliminary}
\label{sec:Preliminary}

Here, we detail the notion of characteristic graphs and their entropy in the context of source compression. 
We recall that the below proofs use the notation given in Section~\ref{sec:org}.

\subsection{Distributed source compression, and communication cost}
\label{subsec:Preliminary-source}
Given statistically dependent, finite-alphabet, i.i.d. random sequences ${\bf X}_1,\,{\bf X}_2,\dots, {\bf X}_N$ where ${\bf X}_i\in\mathbb{F}_q^{|\mathcal{Z}_i|\times n}$ for $i\in\Omega$, the Slepian-Wolf theorem gives a theoretical lower bound for the lossless coding rate of distributed servers, in the limit as $n$ goes to infinity. Denoting by $R_i$ the encoding rate of server $i\in \Omega$, the sum-rate (or communication cost) for distributed source compression is given by 
\begin{align}
\label{communication_cost_SW}
\sum\limits_{i\in \mathcal{S}}R_i\geq H( X_{\mathcal{S}}\, \vert\, X_{\mathcal{S}^c}),\quad \mathcal{S}\subseteq \Omega\ ,
\end{align}
where $\mathcal{S}$ denotes the indices of a subset of servers, $\mathcal{S}^c=\Omega\backslash \mathcal{S}$ its complement, and $X_{\mathcal{S}}=\{X_i,\, i\in \mathcal{S}\}$.

We recall that in the case of distributed source compression, given by the coding theorem of Slepian-Wolf \cite{SlepWolf1973}, the encoder mappings specify the bin indices for the server sequences ${\bf X}_i$. 
The bin index is such that every bin of each n-vector ${\bf X}_i$ of server $i\in\Omega$ is randomly drawn under the uniform distribution across the set $\{0,1,\dots, 2^{nR_i}-1\}$ of $2^{nR_i}$ bins. Transmission of server $i\in\Omega$ is $e_{X_i}({\bf X}_i)$, where $e_{X_i}: {\bf X}_i\to \{0,1,\dots, 2^{nR_i}-1\}$ is the encoding function of $i\in\Omega$ onto bins. 
The total number of symbols in $e_{X_i}({\bf X}_i)$ is $T_i=H(e_{X_i}({\bf X}_i))$. This value corresponds to the aggregate number of symbols in the transmitted subfunctions from the server. 
Hence, the communication cost (rate) of $i\in\Omega$, for a sufficiently large $n$, satisfies
\begin{align}
\label{communication_cost_source_compression}
R_i=\frac{T_i}{L}=\frac{H(e_{X_i}({\bf X}_i))}{n}\geq H(X_i) \ ,
\end{align}
where the cost can be further reduced via a more efficient mapping $e_{X_i}({\bf X}_i)$ if $W_k$, $k\in[K]$ are correlated.

\subsection{Characteristic graphs, distributed functional compression, and communication cost}
\label{subsec:Preliminary-graph}

In this section, we provide a summary of key graph theoretic points devised by  K{\"o}rner \cite{Korner1973}, and studied by Alon and Orlitsky \cite{AlonOrlit1996}, Orlitsky and Roche \cite{OrlRoc2001} to understand the fundamental limits of distributed computation.  

Let us consider the canonical scenario with two servers, storing $X_1$ and $X_2$, respectively. The user requests a bivariate function $F(X_1,\,X_2)$ that could be linearly separable or in general non-linear. Associated with the source pair $(X_1, X_2)$ is a characteristic graph $G$, as defined by Witsenhausen \cite{witsenhausen1976}.   
We denote by $G_{X_1}=(V_{G_{X_1}},\,E_{G_{X_1}})$ the characteristic graph server one builds (server two similarly builds $G_{X_2}$) for computing\footnote{We detail the compression problem for the simultaneous computation of a set of requested functions in Appendix~\ref{app-graph-encode}.} $F(X_1,\,X_2)$, determined as a function of $X_1$, $X_2$, and $F$, where $V_{G_{X_1}}=\mathcal{X}_1$ and an edge $(x_1^1, x_1^2)\in E_{G_{X_1}}$ if and only if there exists a $x_2^1 \in \mathcal{X}_2$ such that $P_{X_1,X_2}(x^1_1, x^1_2)\cdot P_{X_1,X_2}(x^2_1, x^1_2) > 0 $ and $F(x^1_1, x^1_2)\neq F(x^2_1, x^1_2)$. 
Note that the idea of building $G_{X_1}$ can also be generalized to multivariate functions, $F(X_{\Omega})$ where $\Omega=[N]$ for $N>2$ \cite{FeiMed2014}. 
In this paper, we only consider vertex colorings. A valid coloring of a graph $G_{X_1}$ is such that each vertex of $G_{X_1}$ is assigned a color (code) such that adjacent vertices receive disjoint colors (codes). Vertices that are not connected can be assigned to the same or different colors. 
The chromatic number $\chi(G_{X_1})$ of a graph $G_{X_1}$ is the minimum number of colors needed to have a valid coloring of $G_{X_1}$ \cite{malak2022fractional,malak2023weighted,salehi2023achievable}.

\begin{defi}{(Characteristic graph entropy \cite{Korner1973,AlonOrlit1996}.)}
Given a random variable $X_1$ with characteristic graph $G_{X_1} = (V_{X_1} , E_{X_1} )$ for computing function $f(X_1, X_2)$, the entropy of the characteristic graph is expressed as 
\begin{align}
\label{graph_entropy}
H_{G_{X_1}}(X_1)= \min\limits_{X_1\in U_1\in S(G_{X_1})} I(X_1; U_1)\ ,
\end{align}
where $S(G_{X_1})$ is the set of all MISs of $G_{X_1}$, where an MIS is not a subset of any other independent set, where an independent set of a graph is a set of its vertices in which no two vertices are adjacent \cite{moon1965cliques}. 
Notation $X_1 \in U_1 \in S(G_{X_1})$ means that the minimization is over all distributions $P_{U_1,X_1}(u_1, x_1)$ such that $P_{U_1,X_1}(u_1,x_1) > 0$ implies $x_1 \in u_1$, where $U_1$ is an MIS of $G_{x_1}$.

Similarly, the conditional graph entropy for $X_1$ with characteristic graph $G_{X_1}$ for computing $f(X_1, X_2)$ given $X_2$ as side information is defined in \cite{OrlRoc2001} using the notation $U_1-X_1-X_2$ that indicates a Markov chain:
\begin{align}
\label{conditional_graph_entropy}
H_{G_{X_1}}(X_1\,|\,X_2)= \min\limits_{\overset{X_1\in U_1\in S(G_{X_1})}{U_1-X_1-X_2}} I(X_1; U_1\,|\,X_2)\ .
\end{align}
\end{defi}

The Markov chain relation in (\ref{conditional_graph_entropy}) implies that $H_{G_{X_1}}(X_1\,|\,X_2)\leq H_{G_{X_1}}(X_1)$ \cite[Ch. 2]{CovTho1991}. 
In (\ref{conditional_graph_entropy}), the goal is to determine the equivalence classes $U_1$ of $x_1^i\in X_1$ that have the same function outcome $\forall x_2^1\in \mathcal{X}_2$ such that $P_{X_1,X_2}(x_1^i,x_2^1)>0$. 
We next consider an example to clarify the distinction between characteristic graph entropy, $H_{G_{X_1}}(X_1)$ and entropy of a conditional characteristic graph, or conditional graph entropy, $H_{G_{X_1}}(X_1\,\vert\, X_2)$.

\begin{ex}(Characteristic graph entropy of ternary random variables \cite[Examples 1-2]{OrlRoc2001}.) In this example, we first investigate the characteristic graph entropy $H_{G_{X_1}}(X_1)$ and the conditional graph entropy $H_{G_{X_1}}(X_1\,|\,X_2)$.

\begin{enumerate}%
\item  
Let $P_{X_1}$ be a uniform PMF over the set $\{1,2,3\}$. Assume that $G_X$ has only one edge, i.e., $E_{X_1}=\{(1,3)\}$. Hence, the set of MISs is given as $S(G_{X_1})=\{\{1,2\},\{2,3\}\}$.

To determine the entropy of a characteristic graph, i.e.,  $H_{G_{X_1}}(X_1)$, from (\ref{graph_entropy}), our objective is to minimize $I(X_1;U_1)$, which is a convex function of $\mathbb{P}(U_1 \,\vert\, X_1)$. Hence, $I(X_1;U_1)$ is minimized when the conditional distribution of $\mathbb{P}(U_1\,\vert\, X_1)$ is selected as $\mathbb{P}(U_1=\{1,2\} \,\vert\, X_1=1)=1$, $\mathbb{P}(U_1=\{2,3\} \,\vert\, X_1=3)=1$, and~$\mathbb{P}(U_1=\{1,2\} \,\vert\, X_1=2)=\mathbb{P}(U_1=\{2,3\} \,\vert\, X_1=2)=1/2$. As a result of this PMF, we have 
\begin{align}
H_{G_{X_1}}(X_1)=H(U_1)-H(U_1\,\vert\,X_1)=1-\frac{1}{3}=\frac{2}{3} \ .
\end{align}
\item Let $P_{X_1,X_2}$ be a uniform PMF over the set $\{(x_1,x_2): \, x_1,\, x_2\in\{1,2,3\},\, x_1\neq x_2\}$ and $E_{X_1}=\{(1,3)\}$. Note that $H(X_1|X_2)=1$ given the joint PMF. 
To determine the conditional characteristic graph entropy, i.e., $H_{G_{X_1}}(X_1\,\vert\, X_2)$, using (\ref{conditional_graph_entropy}), our objective is to minimize $I(X_1; U_1|X_2)$, which is convex in $\mathbb{P}(U_1\,\vert\, X_1)$. Hence, $I(X_1; U_1\,|\,X_2)$ %
is minimized when %
$\mathbb{P}(U_1\,\vert\, X_1)$ is selected as $\mathbb{P}(U_1=\{1,2\} \,\vert\, X_1=1)=\mathbb{P}(U_1=\{2,3\} \,\vert\, X_1=3)=1$, and $\mathbb{P}(U_1=\{1,2\}\,\vert\, X_1=2)=\mathbb{P}(U_1=\{2,3\}\,\vert\, X_1=2)=1/2$. Hence, 
we obtain 
\begin{align}
H(U_1\,\vert\, X_2)&=\frac{1}{3} H(U_1\,\vert\, X_1\in\{2,3\})+\frac{1}{3}H(U_1\,\vert\, X_1\in\{1,3\})+\frac{1}{3}H(U_1\,\vert\, X_1\in\{1,2\})\nonumber\\
&=\frac{1}{3}h\Big(\frac{1}{4}\Big)+\frac{1}{3}+\frac{1}{3}h\Big(\frac{1}{4}\Big) \ ,
\end{align}
which yields, using $U_1-X_1-X_2$, that 
\begin{align}
H_{G_{X_1}}(X_1\,\vert\,X_2)&=H(U_1\,\vert\, X_2)-H(U_1\,\vert\, X_1,\,X_2)=H(U_1\,\vert\, X_2)-H(U_1\,\vert\, X_1)\nonumber\\
&=\Big[\frac{1}{3}h\Big(\frac{1}{4}\Big)+\frac{1}{3}+\frac{1}{3}h\Big(\frac{1}{4}\Big)\Big]-\frac{1}{3}=\frac{2}{3}h\Big(\frac{1}{4}\Big) \ .
\end{align}
\end{enumerate}

\end{ex}

\begin{defi}\label{chromatic_entropy}
(Chromatic entropy  \cite{AlonOrlit1996}.) The chromatic entropy of a graph $G_{X_1}$
is defined as
\begin{align}
 H^{\chi}_{G_{X_1}}(X_1) = \min\limits_{c_{G_{X_1}}} H(c_{G_{X_1}}(X_1)),   
\end{align}
where the minimization is over the set of colorings such that $c_{G_{X_1}}$ is a valid coloring of $G_{X_1}$.
\end{defi}

Let $G^n_{{\bf X}_1}=(V_{X_1}^n , E_{X_1}^n )$ be the $n$-th OR power of a graph $G_{X_1}$ for the source sequence ${\bf X}_1$ to compress $F({\bf X}_1,\,{\bf X}_2)$. In this OR power graph, $V_{X_1}^n = \mathcal{X}_1^n$ and $({\bf x}_1^1,{\bf x}_1^2) \in E_{X_1}^n$, where ${\bf x}_1^1=(x_{11}^1,x_{12}^1,\dots,x_{1n}^1)$ and similarly for ${\bf x}_1^2$, when there exists at least one coordinate $l\in[n]$ such that $(x_{1l}^1 ,x_{1l}^2 ) \in E_{X_1}$. We denote a coloring of $G^n_{{\bf X}_1}$ by $c_{{G^n_{{\bf X}_1}}}({\bf X}_1)$. 
The encoding function at server one is a mapping from ${\bf X}_1$ to the colors $c_{G^n_{{\bf X}_1}}({\bf X}_1)$ of the characteristic graph $G^n_{{\bf X}_1}$ for computing $F({\bf X}_1,\,{\bf X}_2)$.  
In other words, $c_{G^n_{{\bf X}_1}}({\bf X}_1)$ specifies the color classes of ${\bf X}_1$ such that each color class forms an independent set that induces the same function outcome.

Using Definion~\ref{chromatic_entropy}, we can determine the chromatic entropy of graph $G^n_{{\bf X}_1}$ as
\begin{align}
\label{chromatic_entropy_expression}
H^{\chi}_{G^n_{{\bf X}_1}}({\bf X}_1)= \min\limits_{c_{{G^n_{{\bf X}_1}}}} H(c_{{G^n_{{\bf X}_1}}}({\bf X}_1)) \ .
\end{align}

In \cite{Korner1973}, K\"orner has shown the relation between the chromatic and graph entropies, which we detail next.
\begin{theo}{\bf (Chromatic entropy versus graph entropy \cite{Korner1973}.)} The following relation holds between the characteristic graph entropy and the chromatic entropy of graph $G^n_{{\bf X}_1}$ in the limit of large $n$:
\begin{align}
\label{Korners_graph_entropy}
H_{G_{X_1}}(X_1)&=\lim\limits_{n\to\infty} \frac{1}{n}H^{\chi}_{G^n_{{\bf X}_1}}({\bf X}_1)\ .
\end{align}
\end{theo}

Similarly, from (\ref{chromatic_entropy_expression}) and (\ref{Korners_graph_entropy}), the conditional graph entropy of $X_1$ given $X_2$ is given as
\begin{align}
H_{G_{X_1}}(X_1\, \vert\, X_2)=\lim\limits_{n\to\infty} \,\,\min\limits_{c_{G^n_{{\bf X}_1}},\, c_{G^n_{{\bf X}_2}}} \frac{1}{n} H(c_{G^n_{{\bf X}_1}}({\bf X}_1)\, \vert \, c_{G^n_{{\bf X}_2}}({\bf X}_2)) \ .   
\end{align}

\subsubsection{\bf A characteristic graph-based encoding framework for simultaneously computing a set of functions}
\label{app-graph-encode}

The user demands a set of functions $\{F_j(X_{\Omega})\}_{j\in[K_c]}\in\mathbb{R}^{K_c}$ that are possibly non-linear in the subfunctions.  
In our proposed framework, for the distributed computing of these functions, we leverage characteristic graphs that can capture the structure of subfunctions. To determine the achievable rate of distributed lossless functional compression, we determine the colorings of these graphs and evaluate the entropy of such colorings. In the case of $K_c>1$ functions, let $G_{X_i,j}=(V_{X_i},E_{X_i,j})$ be the characteristic graph that server $i\in\Omega$ builds for computing  function $j\in[K_c]$. The graphs $\{G_{X_i,j}\}_{j\in[K_c]}$ are on the same vertex set.  

Union graphs for simultaneously computing a set of functions with side information have been considered in \cite{FeiMed2014}, using {\emph{multi-functional characteristic graphs}}. A multi-functional characteristic graph is an OR function of individual characteristic graphs for different functions \cite[Definition 45]{FeiMed2014}. 
To that end, server $i\in\Omega$ creates a union of graphs on the same set of vertices $V_{X_i}$ with a set of edges $E^{\cup}_{X_i}$, which satisfy 
\begin{align}
\label{union_graph}
G^{\cup}_{X_i}=\bigcup_{j\in[K_c]} G_{X_i,j}=(V_{X_i},E^{\cup}_{X_i}) \ ,\quad E^{\cup}_{X_i}=\bigcup_{j\in[K_c]} E_{X_i,j} \ .
\end{align}
In other words, we need to distinguish the outcomes $x_i^1$ and $x_i^2$ of server $X_i$ if there exists at least one function $F_j(x_{\Omega}),\,j\in [K_c]$ out of $K_c$ functions such that $F_j(x_i^1,x_{\Omega\backslash_i}^1)\neq F_j(x_i^2,x_{\Omega\backslash_i}^1)$, for some $P_{X_{\Omega}}(x_i^1,x_{\Omega\backslash_i}^1)\cdot P_{X_{\Omega}}(x_i^2,x_{\Omega\backslash_i}^1)>0$ given $x_{\Omega\backslash_i}^1\in X_{\Omega\backslash i}$. 
The server then compresses the union $G^{\cup}_{X_i}$ by exploiting (\ref{chromatic_entropy_expression}) and (\ref{Korners_graph_entropy}). 

In the special case when the number of demanded functions $K_c$ is large 
(or tends to infinity), such that the union of all subspaces spanned by the independent sets of each $G_{X_i,j}$, $j\in [K_c]$ is the same as the subspace spanned by $\mathcal{X}_i$, MISs of $G^{\cup}_{X_i}$ 
in (\ref{union_graph}) for server $i\in\Omega$ become singletons, rendering $G^{\cup}_{X_i}$ a complete graph. In this case, the problem boils down to the paradigm of {\emph{distributed source compression}} (see Appendix~\ref{subsec:Preliminary-source}).

\subsubsection{\bf Distributed functional compression}
\label{app-prelim-Distributed functional compression}

The fundamental limit of functional compression has been given by K\"orner \cite{Korner1973}. 
Given ${\bf X}_i\in \mathbb{F}_q^{|\mathcal{Z}_i|\times n}$ for server $i\in\Omega$, the encoding function $e_{X_i}$ specifies MISs given by the valid colorings $c_{G^n_{{\bf X}_i}}({\bf X}_i)$. 
Let the number of symbols in ${\bf Z}_i=g_i({\bf X}_i)=e_{X_i}(c_{G^n_{{\bf X}_i}}({\bf X}_i))$ be $T_i$ for server $i\in\Omega$. Hence, the communication cost of server $i$, as $n\to\infty$, is given by  (\ref{communication_cost_functional_compression}). 

Defining $G_{X_{\mathcal{S}}}=[G_{X_i}]_{i\in\mathcal{S}}$ for a given subset $\mathcal{S}\subseteq\Omega$ chosen to guarantee distributed computation of $F(X_{\Omega})$, i.e., $|\mathcal{S}|\geq N_r$, the sum-rate of servers for distributed lossless functional compression for computing $F({\bf X}_{\Omega})=\{F(X_{1l},\,X_{2l},\dots, X_{Nl})\}_{l=1}^n$ equals
\begin{align}
\label{communication_cost_Korner}
R_{\rm ach}=\sum\limits_{i\in \mathcal{S}}R_i\geq H_{G_{X_{\mathcal{S}}}}( X_{\mathcal{S}}\, \vert\, Z_{\mathcal{S}^c}),\quad \mathcal{S}\subseteq \Omega\ ,
\end{align}
where $H_{G_{X_{\mathcal{S}}}}(X_{\mathcal{S}})$ is the {\emph{joint graph entropy}} of $\mathcal{S}\subseteq\Omega$, and it is defined as \cite[Definition~30]{FeiMed2014}:
\begin{align}
\label{graph_entropy_general}
H_{G_{X_{\mathcal{S}}}}(X_{\mathcal{S}})=\lim\limits_{n\to\infty} \,\,\min\limits_{\big\{c_{G^n_{{\bf X}_i}}\big\}_{i\in\mathcal{S}}} \frac{1}{n} H(c_{G^n_{{\bf X}_i}}({\bf X}_i),\, i\in\mathcal{S})\ , 
\end{align}
where $c_{G^n_{{\bf X}_i}}({\bf X}_i)$ is the coloring of the $n$-th power graph $G^n_{{\bf X}_i}$ that $i\in\Omega$ builds for computing $f({\bf X}_{\Omega})$ \cite{FeiMed2014}. 

Similarly, exploiting \cite[Definition~31]{FeiMed2014}, the {\emph{conditional graph entropy}} of the servers is given as%
\begin{align}
\label{conditional_graph_entropy_general} 
H_{G_{X_{\mathcal{S}}}}(X_{\mathcal{S}}\, \vert\, Z_{\mathcal{S}^c})=\lim\limits_{n\to\infty} \,\,\min\limits_{\big\{c_{G^n_{{\bf X}_i}}\big\}_{i\in\Omega}} \frac{1}{n} H(c_{G^n_{{\bf X}_i}}({\bf X}_i),\, i\in\mathcal{S}\, \vert \, e_{X_i}(c_{G^n_{{\bf X}_i}}({\bf X}_i)),\, i\in\mathcal{S}^c)\ . 
\end{align}

Using (\ref{union_graph}) we jointly capture the structures of the set of demanded functions. Hence, this enables us to provide a refined communication cost model in (\ref{communication_cost_functional_compression}) versus the characterizations as a function of $K_c$, see e.g., \cite{wan2021distributed,9521491, huang2023fundamental}.

%

\section{Proofs of Main Results}
\label{sec:Proofs}

\subsection{Proof of Theorem~\ref{theo_cyclic_placement_general_source-general_function_graph-rate_UB}}
\label{Proof_theo_cyclic_placement_general_source-general_function_graph-rate_UB}

Consider the general topology, $\mathcal{T}(N, K, K_c, M, N_r)$, under general placement of datasets, and for a set of $K_c$ general functions $\{f_j(W_{\mathcal{K}})\}_{j\in[K_c]}$ requested by the user, and under general jointly distributed dataset models, including non-uniform inputs and allowing correlations across datasets.

We note that server $i\in\Omega$ builds a characteristic graph\footnote{The characteristic graph-based approach is valid  provided that each subfunction $W_k$, $k\in\mathcal{K}$ contained in $X_i
=W_{\mathcal{Z}_i}$ is defined over a $q$-ary field such that  
$q\geq 2$, to ensure that the union graph $G^{\cup}_{X_i}$, $i\in\Omega$ (or $G_{X_i,j}$, $j\in[Kc]$ each) has more than one vertex.} $G_{X_i,j}$ for distributed lossless computing of $f_j(W_{\mathcal{K}})$, $j\in [K_c]$.  
Similarly, server $i\in\Omega$ builds a union characteristic graph for computing $\{f_j(W_{\mathcal{K}})\}_{j\in[K_c]}$. We denote by $G^{\cup}_{X_i}=(V_{X_i},E_{X_i})=\bigcup_{j\in[K_c]} G_{X_i,j}$ the {\em union characteristic graph}, given as in (\ref{union_graph}). In the description of $G^{\cup}_{X_i}$, the set $V_{X_i}$ is the support set of $X_i$, i.e., $V_{X_i}=\mathcal{X}_i$, and $E_{X_i}$ is the union of edges, i.e., $E_{X_i}=\bigcup\nolimits_{j\in[K_c]} E_{X_i,j}$, where $E_{X_i,j}$ denotes the set of edges in $G_{X_i,j}$, which is the characteristic graph the server builds for distributed lossless computing $f_j(W_{\mathcal{K}})$ for a given function $j\in[K_c]$. 

To compute the set of demanded functions $\{f_j(W_{\mathcal{K}})\}_{j\in[K_c]}$, we assume that server $i\in\Omega$ can use a codebook of functions denoted by $\mathcal{C}_i$ such that  $\mathcal{C}_i\ni g_i$, where the user can compute its demanded functions using the set of transmitted information $\{g_i(X_i)\}_{i\in \mathcal{S}}$ provided from any set of $|\mathcal{S}|=N_r$ servers. 
More specifically, server $i\in\Omega$ chooses a function $g_i\in\mathcal{C}_i$ to encode $X_i$. Note that $g_i$ represents, in the context of encoding characteristic graphs, the mapping from $X_i$ to a valid coloring $c_{G_{X_i}}(X_i)$. We denote by ${\bf Z}_i=g_i({\bf X}_i)=e_{X_i}(c_{G^n_{{\bf X}_i}}({\bf X}_i))$ the color encoding performed by server $i\in\Omega$ for the length $n$ realization of $X_i$, denoted by ${\bf X}_i$. For convenience, we use the following shorthand notation to represent the transmitted information from the server: 
\begin{align}
\label{transmitted_info}
Z_i=g_i(X_i) \ , \quad i\in\Omega \ .
\end{align}

Combining the notions of the union graph in (\ref{union_graph}) and the encodings of the individual servers given in (\ref{transmitted_info}), the rate $R_i$ needed from server $i\in\Omega$ for meeting the user demand is upper bounded by the cost of the best encoding that minimizes the rate of information transmission from the respective server. Equivalently, 
\begin{align}
\label{individual_rate_UB}
R_i \geq \min\limits_{Z_i=g_i(X_i)\, : \, g_i\in \mathcal{C}_i} H_{G^{\cup}_{X_i}}(X_i) \ ,
\end{align}
where equality is achievable in (\ref{individual_rate_UB}).  
Because the user can recover the desired functions using any set of $N_r$ servers, the achievable sum rate is upper bounded by 
\begin{align}
\label{eq:general_communication_rate_upper_bound_app}
R_{\rm ach} \leq \sum\limits_{i=1}^{N_r}\, \min\limits_{Z_i=g_i(X_i)\, : \, g_i\in \mathcal{C}_i} H_{G^{\cup}_{X_i}}(X_i) \ .
\end{align}

\subsection{Proof of Proposition~\ref{prop_cyclic_placement_iid_uniform_source-linear_function_graph-rate_UB}}
\label{Proof_prop_cyclic_placement_iid_uniform_source-linear_function_graph-rate_UB}
For the multi-server, multi-function distributed computing architecture, this proposition restricts the demand to be a set of linearly separable functions, given as in (\ref{linearly_separable_functions}). 
Given the recovery threshold $N_r$, it holds that 
\begin{align}
\label{achievable_rate_linear_function_independent_set}
R_{\rm ach}\leq \sum\limits_{i=1}^{N_r}\, \min\limits_{Z_i=g_i(X_i)\, : \, g_i\in \mathcal{C}_i} H_{G^{\cup}_{X_i}}(X_i)&=\sum\limits_{i=1}^{N_r} \, \min\limits_{Z_i\, : \, g_i\in \mathcal{C}_i}\, \min\limits_{X_i\in U_i\in S(G^{\cup}_{X_i})} I(X_i; U_i)\\
\label{mutual_info_definition_indep_set_for_Kc_functions}
&=\sum\limits_{i=1}^{N_r} \, \Big[H(W_{(i-1)\Delta+1}^{(i-1)\Delta+M})-H\Big(W_{(i-1)\Delta+1}^{(i-1)\Delta+M}\, \Big\vert \, Z_i \Big) \Big]\\
\label{M_iid_W_k_variables}
&=\sum\limits_{i=1}^{N_r} \Big[M-\Big(M-H(Z_i)\Big)\Big]
=\sum\limits_{i=1}^{N_r} H(Z_i) \ ,
\end{align}
where in (\ref{achievable_rate_linear_function_independent_set}), we used the identity $H_{G^{\cup}_{X_i}}(X_i)=\min\limits_{X_i\in U_i\in S(G^{\cup}_{X_i})} I(X_i; U_i)$. Furthermore, if the codebook $\mathcal{C}_i$ is restricted to linear combinations of subfunctions, $Z_i$ is given by the following set of linear equations: 
\begin{align}
\label{linear_encoding}
Z_i=g_i(X_i)=\Big\{\sum\limits_{k=(i-1)\Delta+1}^{(i-1)\Delta+M} \alpha_k^{(l)} W_k\ ,\,l\in [K_c]\Big\}\ .
\end{align}
In other words, $Z_i$, $i\in[N_r]$, is a vector-valued function. 
Note that each server contributes to determining the set of linearly separable functions $\{f_j(W_{\mathcal{K}})\ , \ j\in[K_c]\}$ of datasets, given as in (\ref{linearly_separable_functions}), in a distributed manner. 
Hence, each independent set $U_i\in S(G^{\cup}_{X_i})$, with  $S(G^{\cup}_{X_i})$ denoting the set of MISs of $X_i$, of $X_i$ is captured by the linear functions of $\{W_k\}_{k\in [(i-1)\Delta+1:(i-1)\Delta+M]}$, i.e., each $U_i\in  S(G^{\cup}_{X_i})$ is determined by (\ref{linear_encoding}). 
Hence, the user can recover the requested functions by linearly combining the transmissions of the $N_r$ servers:
\begin{align}
\label{linear_combination_transmissions_of_linear_encodings}
f_j(W_{\mathcal{K}})=\sum\limits_{i=1}^{N_r} \beta_{ji} Z_i=\sum\limits_{i=1}^{N_r} \beta_{ji} g_i(X_i)=\sum\limits_{k=1}^{K}\gamma_{jk} W_k\ ,\quad j\in [K_c] \ . 
\end{align}
In (\ref{mutual_info_definition_indep_set_for_Kc_functions}), we use the definition of mutual information, $I(X_i; U_i)=H(X_i)-H(X_i\,\vert\, U_i)$, where given $i\in [N_r]$ and $\Delta=\frac{K}{N}$, it holds under cyclic placement that 
\begin{align}
X_i=W_{(i-1)\Delta+1}^{(i-1)\Delta+M}=W_{(i-1)\Delta+1},W_{(i-1)\Delta+2},\dots,W_{(i-1)\Delta+M} \ ,
\end{align}
and
$\alpha_k^{(l)}$ are the 
coefficients for computing function $l\in[K_c]$. 
In (\ref{M_iid_W_k_variables}), we used that $W_k$ is uniform over $\mathbb{F}_q$ and i.i.d. across $k\in [K]$, and rewrote the conditional entropy expression such that
\begin{align}
H\Big(W_{(i-1)\Delta+1}^{(i-1)\Delta+M}\, \Big\vert \, Z_i\, \Big)=H\Big(W_{(i-1)\Delta+1}^{(i-1)\Delta+M}\, , \, Z_i\, \Big)-H(Z_i)
\overset{(a)}{=}H\Big(W_{(i-1)\Delta+1}^{(i-1)\Delta+M}\Big)-H(Z_i) \ , 
\end{align}
where $(a)$ follows from that $Z_i$ is a function of $W_{(i-1)\Delta+1}^{(i-1)\Delta+M}$.
For a given $l\in[K_c]$ and field size $q$, the relation $\sum\nolimits_{k=(i-1)\Delta+1}^{(i-1)\Delta+M} \alpha_k^{(l)} W_k$ ensures that $G_{X_i}$ has $q$ independent sets where each such  
set $U_i$ contains $q^{M-1}$ different values of $X_i$. 
Exploiting that $W_k$ is i.i.d. and uniform over $\mathbb{F}_q$, each element of $Z_i$ is uniform over $\mathbb{F}_q$. Hence, the achievable sum-rate is upper bounded by 
\begin{align}
\label{achievable_rate_trivialupperbound_linear_function_independent_set}
\sum\limits_{i=1}^{N_r}\, \min\limits_{Z_i\, : \, g_i\in \mathcal{C}_i} H_{G^{\cup}_{X_i}}(X_i) \leq K_c N_r \ .
\end{align} 
Exploiting the cyclic placement model, we can tighten the bound in (\ref{achievable_rate_trivialupperbound_linear_function_independent_set}). Note that server $i=1$ can help recover $M$ subfunctions (at most, i.e., $M$ transmissions needed to recover $M$ subfunctions), and each of servers $i\in[2: N_r]$ can help recover an additional $\Delta$ subfunctions (at most, i.e., $\Delta$ transmissions needed to recover $\Delta$ subfunctions). Hence, the set of servers $[N_r]$ suffices to provide $M+(Nr-1)\Delta=N\Delta=K$ subfunctions and reconstruct any desired function of $W_{\mathcal{K}}$. 
Due to cyclic placement, each $W_k$ is stored in exactly $N-N_r+1$ servers. 
Now, let us consider the following four scenarios:

(i) When $1\leq K_c <\Delta$, it is sufficient for each server to transmit $K_c$ linearly independent combinations of their subfunctions. This leads to resolving $K_c N_r$ linear combinations of $K$ subfunctions from $N_r$ servers that are sufficient to derive the demanded $K_c$ linear functions. Because $K_c N_r<\Delta N_r$, there are $K-K_c N_r>\Delta (N-N_r)=M-\Delta$ unresolved linear combinations of $K$ subfunctions.   

(ii) When $\Delta\leq K_c\leq \Delta N_r$, it is sufficient for each server to transmit at most $\Delta$ linearly independent combinations of their subfunctions. This leads to resolving $\Delta N_r$ linear combinations of $K$ subfunctions and $\Delta (N-N_r)=M-\Delta$ unresolved linear combinations of $K$ subfunctions. 

(iii) When $\Delta N_r<K_c\leq K$, each server needs to transmit at a rate $\frac{K_c}{N_r}$ where $\frac{K_c}{N_r}>\Delta$ and  $\frac{K_c}{N_r}\leq \frac{K}{N_r}=\Delta\big(\frac{N_r+N-N_r}{N_r}\big)=\Delta + \Delta\big(\frac{N-N_r}{N_r}\big)$, which gives the number of linearly independent combinations needed to meet the demand. This yields a sum-rate of $K_c$. The subset of servers may need to provide up to an additional $\Delta(N-N_r)$ linear combinations, and $\Delta\big(\frac{N-N_r}{N_r}\big)$ defines the maximum number of additional linear combinations per server, i.e., the required number of combinations when $K_c=K$. 

(iv) When $K<K_c$, it is easy to note that since any $K$ linearly independent equation in (\ref{linear_combination_transmissions_of_linear_encodings}) suffices to recover $W_{\mathcal{K}}$, the sum-rate $K$ is achievable. 

From (i)-(iv), we obtain the following upper bound on the achievable sum-rate:
\begin{align}
\label{achievable_rate_tighterupperbound_linearly_separable_different_ranges_Kc}
\sum\limits_{i=1}^{N_r}\, \min\limits_{Z_i\, : \, g_i\in \mathcal{C}_i} H_{G^{\cup}_{X_i}}(X_i)=\begin{cases}
K_c N_r\ ,\quad 1\leq K_c <\Delta \ ,\\
\Delta N_r\ ,\quad \Delta \leq K_c \leq \Delta Nr \ ,\\
K_c \ ,\quad \Delta Nr<K_c\leq K \ , \\
K \ ,\quad K<K_c \ ,
\end{cases}
\end{align}
where it is easy to note that (\ref{achievable_rate_tighterupperbound_linearly_separable_different_ranges_Kc}) matches the communication cost in \cite[Theorem 2]{wan2021distributed}. 
The i.i.d. distribution assumption for $W_k$, ensures that this result holds for any $q\geq 2$.

\subsection{Proof of Proposition~\ref{prop:general_placement__correlated_Boolean_function}}
\label{proof-prop:general_placement__correlated_Boolean_function}

Similarly as in Theorem~\ref{theo_cyclic_placement_general_source-general_function_graph-rate_UB}, we let  $G^{\cup}_{X_i}=\bigcup_{j\in[K_c]} G_{X_i,j}$ denote the {\em union characteristic graph} that server $i\in\Omega$ builds for computing $\{f_j(W_{\mathcal{K}})\}_{j\in[K_c]}$. 
Note that given $W_{\mathcal{Z}_i}=W_{1+(i-1)\Delta},W_{1+(i-1)\Delta+1},\dots, W_{M+(i-1)\Delta}$, the support set of server $i\in \Omega$ has a cardinality of $\mathcal{X}_i=2^M$. 
Because the user demand is a collection of Boolean functions, in this scenario, each server $i\in\Omega$ builds a graph with $2$ independent sets at most, denoted by $s_0(G^{\cup}_{X_i})$ and $s_1(G^{\cup}_{X_i})$, yielding the function values $Z_i=0$ and $Z_i=1$, respectively.

Given the recovery threshold $N_r$, any subset $\mathcal{S}$ of servers with $|\mathcal{S}|=N_r$ stores the set 
$\mathcal{K}$, which is sufficient to compute the demanded functions. 
Given server $i\in\Omega$, consider the set of all $w_{\mathcal{Z}_i}\in W_{\mathcal{Z}_i}$  that satisfies
\begin{align}
\label{Boolean_indep_set}
f(w_{\mathcal{Z}_i},w_{\mathcal{Z}_{\mathcal{S}}\backslash \mathcal{Z}_i})=1 \ ,\quad \forall w_{\mathcal{Z}_{\mathcal{S}}\backslash \mathcal{Z}_i}\in\{0,1\}^{|\mathcal{K}\backslash \mathcal{Z}_i|} \ ,
\end{align}
where notation $w_{\mathcal{Z}_{\mathcal{S}}\backslash \mathcal{Z}_i}$ denotes the dataset values for the set of datasets stored in the subset of servers $\mathcal{S}\backslash i$. 
Note in general that $K_n(\mathcal{S})=|\mathcal{Z}_{\mathcal{S}}|=\big|\bigcup\nolimits_{i\in \mathcal{S}} \mathcal{Z}_i\big|$. In the case of cyclic placement based on  (\ref{cyclic_Zi}), out of the set of all datasets $\mathcal{K}$, there are $\Delta$ datasets that belong exclusively to server $i\in\Omega$. In this case, $|\mathcal{K}\backslash \mathcal{Z}_i|=K-\Delta$. 

Note that (\ref{Boolean_indep_set}) captures the independent set $ s_1(G^{\cup}_{X_i})\ni w_{\mathcal{Z}_i}$.  
Equivalently, the set of dataset values %
$W_{\mathcal{Z}_i}$ that lands in $s_1(G^{\cup}_{X_i})$ of $G^{\cup}_{X_i}$, yields $Z_i=1$. 
The transmitted information takes the value $Z_i=1$ with a probability 
\begin{align}
\mathbb{P}(Z_i=1)=\mathbb{P}(W_{\mathcal{Z}_i}\in s_1(G^{\cup}_{X_i}))\ , \quad i\in \Omega \ ,
\end{align}
using which the upper bound on the achievable sum rate can be determined.

\subsection{Proof of Proposition~\ref{prop:multi_shot_multilinear_function}}
\label{proof-prop:multi_shot_multilinear_function}
Recall that $W_k\sim {\rm Bern}(\epsilon)$ are i.i.d. across $k\in [K]$, and each server has a capacity $M=\Delta(N-N_r+1)$. This means that given the number of datasets $K$, each server can compute the product of $\Delta(N-N_r+1)$ subfunctions and hence, the minimum number of servers to evaluate the multi-linear function $f(W_{\mathcal{K}})=\prod\nolimits_{k\in [K]} W_k$ is $N^*=\left\lfloor \frac{N}{N-N_r+1}\right\rfloor$ such that 
given its capacity  $M=|\mathcal{Z}_i|$, %
each server can compute the product of a disjoint set of $M$ subfunctions, i.e., $\prod\nolimits_{k\in \mathcal{Z}_i} W_k$, which operates at a rate of $R_i\geq h(\epsilon_M)$, $i\in\Omega$. 
Exploiting the characteristic graph approach, we build $G_{X_1} = (V_{X_1}\, ,\, E_{X_1})$ for $X_1$, with respect to variables $X_{\Omega}\backslash X_1=X_2,\dots, X_N$ and $f(W_{\mathcal{K}})$, and similarly for other servers to characterize the sum-rate for the computation by evaluating the entropy of each graph.

To evaluate the first term in (\ref{sum_rate_multishot_multilinear}), we choose a total of $N^*$ servers with a disjoint set of subfunctions. We denote the selected set of servers by $\mathcal{N}^*\subseteq \Omega$, and the collective computation rate of these $N^*$ servers, as a function of the conditional graph entropies of these servers, becomes
\begin{align}
\label{sum_rate_cond_graph_entropies}
\sum\limits_{i\in \mathcal{N}^*} R_i&\overset{(a)}{\leq} H_{G_{X_{i_1}}}(X_{i_1})+H_{G_{X_{i_2}}}(X_{i_2}\,\vert\, Z_{i_1})+\dots+H_{G_{X_{i_{N^*}}}}(X_{i_{N^*}}\, \vert\, Z_{i_1},Z_{i_2},\dots, Z_{i_{N^*-1}})\nonumber\\
&\overset{(b)}{=}h(\epsilon_M)+\epsilon_M h(\epsilon_M)+(\epsilon_M)^2 h(\epsilon_M)+\cdots+(\epsilon_M)^{N^*-1} h(\epsilon_M)\nonumber\\
&\overset{(c)}{=}\frac{1-(\epsilon_M)^{N^*}}{1-\epsilon_M}\cdot h(\epsilon_M) \ ,
\end{align}
where $(a)$ follows from assuming $\mathcal{S}=\{i_1,i_2,\cdots,i_{N^*}\}$ with no loss of generality, and $(b)$ from that the rate of server $i_l\in \mathcal{S}$ is positive only when $\prod\nolimits_{i\in[i_{l-1}]}\prod\nolimits_{k\in\mathcal{Z}_i}W_k=1$, which is true with probability $(\epsilon_M)^{l-1}$. Finally, $(c)$ follows from employing the sum of the terms in the geometric series, i.e., $\sum\nolimits_{l=0}^{N^*-1} (\epsilon_M)^l=\frac{1-(\epsilon_M)^{N^*}}{1-\epsilon_M}$. 
\footnote{While Proposition~\ref{prop:multi_shot_multilinear_function} uses the conditional graph entropies, the statements of Theorem~\ref{theo_cyclic_placement_general_source-general_function_graph-rate_UB} and Proposition~\ref{prop_cyclic_placement_iid_uniform_source-linear_function_graph-rate_UB}, and Proposition~\ref{prop:general_placement__correlated_Boolean_function} do not %
take into account the notion of conditional graph entropies. 
However, as indicated in 
Subsection~\ref{ex:binary_lin_sep} for computing linearly separable functions, and in Subsection~\ref{ex_prop:multi_shot_multilinear_function} for computing multi-linear functions, respectively, we used the conditional entropy-based sum rate in (\ref{sum_rate_cond_graph_entropies}) to evaluate and illustrate the achievable gains over \cite{wan2021distributed} and \cite{SlepWolf1973}.}

In the case of $\Delta_N=N-N^*\cdot (N-N_r+1)>0$, the product of $K$ subfunctions cannot be determined by $N^*$ servers and we need %
additional servers $\mathcal{I}^*\in\Omega$ 
to aid the computation and determine the outcome of $f(W_{\mathcal{K}})$ by computing the product of the remaining $\Deltaprod$ subfunctions. 
In other words, if $\Delta_N>0$ and $\prod\nolimits_{i\in\mathcal{S}}\prod\nolimits_{k\in\mathcal{Z}_i}W_k=1$, the $(N^*+1)$-th server determines the outcome of $f(W_{\mathcal{K}})$ by computing the product of subfunctions $W_k\sim {\rm Bern}(\epsilon)$, $k\in \big[N-\Deltaprod+1:N\big]$, that cannot be captured by the previous $N^*$ servers. Hence, the additional rate, given by the second term in (\ref{sum_rate_multishot_multilinear}), is given by the product of the term
\begin{align}
(\epsilon_M)^{N^*}=\mathbb{P}\Big(\prod\limits_{i\in\mathcal{S}}\prod\limits_{k\in\mathcal{Z}_i}W_k=1\Big) \ , 
\end{align}
with %
$1_{\Delta_N>0}$, and %
$h\big(\epsilon_{\Deltaprod}\big)$. 
Combining this rate term with (\ref{sum_rate_cond_graph_entropies}), we prove the statement of the proposition.


\begin{spacing}{0.9}
\bibliographystyle{IEEEtran}
\bibliography{Derya}
\end{spacing}
\end{document}